\documentclass{new_tlp}

\newif\ifarxiv
\arxivtrue

\newif\iftplp

\ifarxiv \usepackage{hyperref} \fi
\usepackage{amsfonts,latexsym,amssymb,color,amsmath}
\usepackage{comment}

\usepackage{stmaryrd}

\ifarxiv \hypersetup{breaklinks} \fi

\newcommand{\mycomment}[1]{}
\newcommand{\mysim}{\sim\!}
\newcommand{\lsem}{\mbox{$\lbrack\!\lbrack$}}
\newcommand{\rsem}{\mbox{$\rbrack\!\rbrack$}}
\newcommand{\ctrue}{\mathsf{true}}
\newcommand{\cfalse}{\mathsf{false}}
\newcommand{\mtrue}{\mathit{true}}
\newcommand{\mfalse}{\mathit{false}}
\newcommand{\mundef}{\mathit{undef}}
\newcommand{\mnot}{\sim\!\!}
\newcommand{\per}{\mbox{{\tt .}}}
\newcommand{\lpa}{\mbox{{\tt (}}}
\newcommand{\rpa}{\mbox{{\tt )}}}

\newtheorem{lemma}{Lemma}
\newtheorem{proposition}{Proposition}
\newtheorem{definition}{Definition}
\newtheorem{theorem}{Theorem}
\newtheorem{corollary}{Corollary}
\newtheorem{example}{Example}

\begin{document}
\title[Approximation Fixpoint Theory and the Well-Founded Semantics of HOLP]{Approximation Fixpoint Theory and the Well-Founded Semantics of Higher-Order Logic Programs}

\author[A. Charalambidis, P. Rondogiannis and I. Symeonidou]
			 {ANGELOS CHARALAMBIDIS\\
				Institute of Informatics and Telecommunications, NCSR ``Demokritos'', Greece \\
				\email{acharal@iit.demokritos.gr}
				\and PANOS RONDOGIANNIS, IOANNA SYMEONIDOU \\
				Department of Informatics and Telecommunications, University of Athens, Greece \\
		    \email{\{prondo,sioanna\}@di.uoa.gr}}

\pagerange{\pageref{firstpage}--\pageref{lastpage}}
\volume{\textbf{10} (3):}
\jdate{March 2002}
\setcounter{page}{1}
\pubyear{2002}

\maketitle

\begin{abstract}
We define a novel, extensional, three-valued semantics for higher-order logic programs
with negation. The new semantics is based on interpreting the types of the source language
as three-valued Fitting-monotonic functions at all levels of the type hierarchy. We prove
that there exists a bijection between such Fitting-monotonic functions and pairs of two-valued-result
functions where the first member of the pair is monotone-antimonotone and the second member is
antimonotone-monotone. By deriving an extension of {\em consistent approximation fixpoint
theory}~\cite{DMT04} and utilizing the above bijection, we define an iterative procedure
that produces for any given higher-order logic program a distinguished extensional model.
We demonstrate that this model is actually a {\em minimal} one. Moreover, we prove that our
construction generalizes the familiar well-founded semantics for classical logic programs,
making in this way our proposal an appealing formulation for capturing the {\em well-founded
semantics for higher-order logic programs}.
\ifarxiv This paper is under consideration for acceptance in TPLP.\fi
\end{abstract}


\section{Introduction}
An intriguing and difficult question regarding logic programming, is whether it can be extended
to a higher-order setting without sacrificing its semantic simplicity and clarity. Research results
in this direction~\cite{Wa91a,Bezem99,CharalambidisHRW13,RS16,RondogiannisS17} strongly suggest
that it is possible to design higher-order logic programming languages that have powerful expressive
capabilities, and which, at the same time, retain all the desirable semantic properties of classical
first-order logic programming. In particular, it has been shown that higher-order logic programming can be
given an {\em extensional semantics}, namely one in which program predicates denote sets.
Under such a semantics one can use standard set-theoretic concepts in order to understand
the meaning of programs and reason about them. For a more detailed discussion of extensionality
and its importance for higher-order logic programming, the interested reader can consult the
discussion in Section 2 of~\cite{RondogiannisS17}.

The above line of research started many years ago by W. W. Wadge~\cite{Wa91a} who considered
{\em positive} higher-order logic programs (i.e., programs without negation in clause bodies).
Wadge argued that if such a program obeys some simple and natural syntactic rules, then it has a unique
{\em minimum} Herbrand model. It is well-known that the {\em minimum model property} is a cornerstone of
the theory of first-order logic programming~\cite{EmdenK76}. In this respect, Wadge's result suggested that it might
be possible to extend all the elegant theory of classical logic programming to the higher-order case.
The results in~\cite{Wa91a} were obtained using standard techniques from denotational semantics
involving {\em continuous} interpretations and Kleene's least fixpoint theorem.
A few years after Wadge's initial result, M. Bezem came to similar conclusions~\cite{Bezem99}
but from a different direction. In particular, Bezem demonstrated that by using a fixpoint
construction on the ground instantiation of the source higher-order program, one can
obtain a model of the original program that satisfies an extensionality condition
defined in~\cite{Bezem99}. Despite their different philosophies, Wadge's and
Bezem's approaches have recently been shown~\cite{CharalambidisRS17} to have close
connections. Apart from the above results, recent work~\cite{CharalambidisHRW13}
has also shown that we can define a sound and complete proof procedure for positive
higher-order logic programs, which generalizes classical SLD-resolution. In other words,
the central results for positive first-order logic programs, generalize
to the higher-order case.

A natural question that arises is whether one can still obtain an extensional semantics
if negation is added to programs. Surprisingly, this question proved harder to resolve.
The first result in this direction was reported in~\cite{CharalambidisER14}, where it was
demonstrated that every higher-order logic program with negation has a minimum extensional
Herbrand model constructed over a logic with an infinite number of truth values. This result was
obtained using domain-theoretic techniques as-well-as an extension of Kleene's fixpoint
theorem that applies to a class of functions that are potentially non-monotonic~\cite{ER15}.
More recently, it was shown in~\cite{RS16} that Bezem's technique for positive programs can also
be extended to apply to higher-order logic programs with negation, provided that it is interpreted
under the same infinite-valued logic used in~\cite{CharalambidisER14}. The above results, although
satisfactory from a mathematical point of view, left open an annoying natural question: ``Is it
possible to define a {\em three-valued} extensional semantics for higher-order logic programs
with negation that generalizes the standard well-founded semantics for classical logic
programs?''.

The above question was recently undertaken in~\cite{RondogiannisS17}. The surprising
result was obtained that if Bezem's approach is interpreted under a three-valued logic,
then the resulting semantics {\em can not be extensional} in the general case. One can
see that similar arguments hold for the technique of~\cite{CharalambidisER14}. Therefore, if we seek
an extensional three-valued semantics for higher-order logic programs with negation,
we need to follow an approach that is radically different from both~\cite{CharalambidisER14}
and~\cite{RondogiannisS17}.

In this paper we undertake exactly the above problem. We demonstrate that we can indeed
define a three-valued extensional semantics for higher-order logic programs with negation,
which generalizes the familiar well-founded semantics of first-order logic programs~\cite{GelderRS91}.
Our results heavily utilize the technique of {\em approximation fixpoint theory}~\cite{DMT00,DMT04},
which proved to be an indispensable tool in our investigation. The main contributions of
the present paper can be outlined as follows:
\begin{itemize}
\item We define the first (to our knowledge) extensional three-valued semantics for higher-order logic programs
      with negation. Our semantics is based on interpreting the predicate types of our language as three-valued Fitting-monotonic
      functions (at all levels of the type hierarchy). We prove that there exists a bijection between such
      Fitting-monotonic functions and pairs of two-valued-result functions of the form $(f_1,f_2)$, where
      $f_1$ is monotone-antimonotone, $f_2$ is antimonotone-monotone, and $f_1\leq f_2$ (these notions
      will be explained in detail in Section~\ref{bijection}).

\item By deriving an extension of {\em consistent approximation fixpoint theory}~\cite{DMT04} and
      utilizing the above bijection, we define an iterative procedure that produces for any given
      higher-order logic program a distinguished extensional model. We prove that this model
      is actually a {\em minimal} one and we demonstrate that our construction generalizes the
      familiar well-founded semantics for classical logic programs. Therefore, we argue that our
      proposal is an appealing formulation for capturing the {\em well-founded semantics for
      higher-order logic programs}, paving in this way the road for a further study of negation
      in higher-order logic programming.
\end{itemize}

The rest of the paper is organized as follows. Section~\ref{intuitive} presents in an intuitive
way the main ideas developed in the paper. Section~\ref{syntax_of_language} introduces
the syntax and Section~\ref{semantics_of_language} the semantics of our source language.
Section~\ref{bijection} demonstrates the bijection between Fitting-monotonic functions
and pairs of monotone-antimonotone and antimonotone-monotone functions. Section~\ref{well_founded}
develops the well-founded semantics of higher-order logic programs with negation, based on
an extension of consistent approximation fixpoint theory. Section~\ref{conclusions}
compares the present work with that of~\cite{CharalambidisER14,RondogiannisS17},
and concludes by identifying some promising research directions. The proofs of most
results of the paper are given in the 
\ifarxiv appendices.\fi

\section{An Intuitive Overview of the Proposed Approach}\label{intuitive}
In this section we describe in an intuitive way the main ideas and results obtained
in the paper. As we have already mentioned, our goal is to derive a generalization
of the well-founded semantics for higher-order logic programs with negation.

We start with our source language ${\cal HOL}$ which, intuitively speaking, allows {\em distinct}
predicate variables (but not predicate constants) to appear in the heads of clauses. This is a
syntactic restriction initially introduced in~\cite{Wa91a}, which has been preserved
and used by all subsequent articles in the area. As an example, consider the following
program (for the moment we use ad-hoc Prolog-like syntax):
\begin{example}\label{prolog-notation}
The program below defines the {\tt subset} relation over two unary predicates {\tt P} and {\tt Q}:
\[
\begin{array}{l}
\mbox{\tt subset(P,Q) $\leftarrow$ $\mysim $ nonsubset(P,Q).}\\
\mbox{\tt nonsubset(P,Q) $\leftarrow$ P(X), $\mysim $ Q(X).}
\end{array}
\]
Intuitively, {\tt P} is a subset of {\tt Q}
if it is not the case that {\tt P} is a non-subset of {\tt Q}; and {\tt P}
is a non-subset of {\tt Q} if there exists some {\tt X} for which {\tt P}
is true while {\tt Q} is false.
\end{example}

The syntax we will introduce in Section~\ref{syntax_of_language} will allow
a more compact notation using $\lambda$-expressions as the bodies of clauses
(see Example~\ref{subset-predicate-example} later in the paper).

We would like, for programs such as the above that are higher-order and use negation,
to devise a three-valued extensional semantics. The key idea when assigning extensional
semantics to {\em positive} higher-order logic programs~\cite{Wa91a,CharalambidisHRW13}
is to interpret the predicate types of the language as monotonic and continuous functions.
This is a well-known idea in the area of denotational semantics~\cite{Tennent} and
is a key assumption for obtaining the least fixpoint semantics for functional programs.
This same idea was used in~\cite{Wa91a,CharalambidisHRW13} for obtaining the
minimum Herbrand model semantics for positive higher-order logic programs. Unfortunately,
this idea breaks down when we consider programs with negation: predicates defined using
negation in clause bodies are not-necessarily monotonic. Non-monotonicity means that a higher-order
predicate may be true of an input relation, but it may be false for a superset of this relation.
For example, consider the predicate {\tt p} below:
\[
\begin{array}{l}
\mbox{\tt p(Q) $\leftarrow$ $\mysim $ Q(a).}
\end{array}
\]
Obviously, {\tt p} is true of the empty relation $\{\,\,\}$ but it is false of the relation
$\{{\tt a}\}$. Notice that the notion of monotonicity we just discussed is usually called
{\em monotonicity with respect to the (standard) truth ordering}.

Fortunately, there is another notion of monotonicity which is obeyed by higher-order
logic programs with negation, namely {\em Fitting-monotonicity} (or {\em monotonicity
with respect to the information ordering})~\cite{Fitting}. Consider the program:
\[
\begin{array}{l}
\mbox{\tt p(Q) $\leftarrow$ $\mysim $ Q(a).}\\
\mbox{\tt r(a) $\leftarrow$ $\mysim $ r(a).}\\
\mbox{\tt s(a).}
\end{array}
\]

Under the standard well-founded semantics for classical (first-order) logic programs, the truth value
assigned to {\tt r(a)} is {\em undefined}; on the other hand, {\tt s(a)} is {\em true}
in the same semantics. In other words, {\tt r} corresponds to the 3-valued relation
$\{({\tt a},\textit{undef})\}$ while {\tt s} to the relation $\{({\tt a},\textit{true})\}$.
Fitting-monotonicity intuitively states that if a relation takes as argument a {\em more
defined} relation, then it returns a more defined result. In our case this means that
we expect the answer to the query {\tt p(r)} to be less defined (alternatively, to
{\em have less information}) than the answer to the query {\tt p(s)} (more specifically,
we expect {\tt p(r)} to be {\em undefined} and {\tt p(s)} to be {\em false}).

Based on the above discussion, we interpret the predicate types of our language
as Fitting-monotonic functions. Then, an {\em interpretation} of a program
is a function that assigns Fitting-monotonic functions to the predicates of the program.
Given a program $\mathsf{P}$, it is straightforward to define its {\em immediate
consequence operator} $\Psi_{\mathsf{P}}$, which, as usual, takes as input a
Herbrand interpretation of the program and returns a new one. It is easy to prove
that $\Psi_{\mathsf{P}}$ is Fitting-monotonic. It is now tempting to assume
that the least fixpoint of $\Psi_{\mathsf{P}}$ with respect to the Fitting ordering,
is the well-founded model that we are looking for. However, this is not the case:
the least fixpoint of $\Psi_{\mathsf{P}}$ is minimal with respect to the
Fitting (i.e., information) ordering, while the well-founded model should be minimal
with respect to the standard truth ordering. In order to get the correct model, we need
a few more steps.

We prove that there exists a bijection between Fitting-monotonic functions and
pairs of functions of the form $(f_1,f_2)$, where  $f_1$ is monotone-antimonotone,
$f_2$ is antimonotone-monotone, and $f_1\leq f_2$ (where $\leq$ corresponds to
the standard truth ordering). A similar bijection is established between
three-valued interpretations and pairs of two-valued-result ones. This bijection
allows us to use the powerful tool of approximation fixpoint theory~\cite{DMT00,DMT04}.
In particular, starting from a pair consisting of an underdefined interpretation
and an overdefined one, and by iterating an appropriate operator, we demonstrate that
we get to a pair of interpretations that is the limit of this sequence. Using our
bijection, we show that this limit pair can be converted to a three-valued interpretation
${\cal M}_{\mathsf{P}}$ which is a three-valued model of our program $\mathsf{P}$ and actually a minimal
one with respect to the standard truth ordering. We argue that this is the well-founded
semantics of $\mathsf{P}$, because its construction is a generalization of the construction
in~\cite{DMT04} for the well-founded semantics of classical logic programs.

\section{The Syntax of the Higher-Order Language ${\cal HOL}$}\label{syntax_of_language}
In this section we introduce ${\cal HOL}$, a higher-order language based on a simple type system that supports two base types: $o$, the boolean
domain, and $\iota$, the domain of individuals (data objects). The composite
types are partitioned into three classes: {\em functional} (assigned to
individual constants, individual variables and function symbols),
{\em predicate} (assigned to predicate constants and variables) and {\em argument}
(assigned to parameters of predicates).
\begin{definition}
A type can either be {\em functional}, {\em predicate}, or {\em argument}, denoted by
$\sigma$, $\pi$ and $\rho$ respectively and defined as:
\begin{align*}
\sigma & :=  \iota \mid \iota \rightarrow \sigma  \\
\pi   & := o \mid \rho \rightarrow \pi  \\
\rho & :=  \iota \mid \pi
\end{align*}
We will use $\tau$ to denote an arbitrary type (either functional, predicate or argument).
\end{definition}

The binary operator $\rightarrow$ is right-associative. A
functional type that is different from $\iota$ will often be written
in the form $\iota^n \rightarrow \iota$, $n\geq 1$ (which stands for
$\iota \rightarrow \iota \rightarrow \cdots \rightarrow \iota$ $(n+1)$-times).
It can be easily seen that every predicate type $\pi$ can be written uniquely
in the form $\rho_1 \rightarrow \cdots \rightarrow \rho_n \rightarrow o$,
$n\geq 0$ (for $n=0$ we assume that $\pi=o$). We now define the alphabet,
the expressions, and the program clauses of ${\cal HOL}$:
\begin{definition}
The \emph{alphabet} of the higher-order language ${\cal HOL}$ consists
of the following:
\begin{enumerate}
\item {\em Predicate variables} of every predicate type $\pi$
      (denoted by capital letters such as
      $\mathsf{P}$ and $\mathsf{Q}$).

\item {\em Predicate constants} of every predicate type $\pi$
      (denoted by lowercase letters such as
      $\mathsf{p}$ and $\mathsf{q}$).

\item {\em Individual variables} of type $\iota$
      (denoted by capital letters such as
      $\mathsf{X}$ and $\mathsf{Y}$).

\item {\em Individual constants} of type $\iota$ (denoted by lowercase
      letters such as $\mathsf{a}$ and $\mathsf{b}$).

\item {\em Function symbols} of every functional type $\sigma \neq \iota$
      (denoted by lowercase letters such as $\mathsf{f}$ and $\mathsf{g}$).

\item The following {\em logical constant symbols}: the constants
      $\cfalse$ and $\ctrue$ of type $o$; the equality  constant $\approx$
      of type $\iota \rightarrow \iota \rightarrow o$; the generalized disjunction
      and conjunction constants $\bigvee_{\pi}$ and $\bigwedge_{\pi}$ of type
      $\pi \rightarrow \pi \rightarrow \pi$, for every predicate type $\pi$;
      the generalized inverse implication constants $\leftarrow_{\pi}$ of type
      $\pi \rightarrow \pi \rightarrow o$, for every predicate type
      $\pi$; the existential quantifier $\exists_{\rho}$ of type
      $(\rho \rightarrow o)\rightarrow o$, for every argument type
      $\rho$; the negation constant $\mnot\,$ of type $o \rightarrow o$.

\item The {\em abstractor} $\lambda$ and the parentheses ``$\mathsf{(}$'' and ``$\mathsf{)}$''.
\end{enumerate}
The set consisting of the predicate variables and the individual
variables of ${\cal HOL}$ will be called the set of {\em argument
variables} of ${\cal HOL}$. Argument variables will be denoted
by $\mathsf{R}$.
\end{definition}
\begin{definition}

The set of {\em expressions} of the higher-order language ${\cal HOL}$
is defined as follows:
\begin{enumerate}
\item Every predicate variable (respectively, predicate constant)
      of type $\pi$ is an expression of type $\pi$; every
      individual variable (respectively, individual constant) of
      type $\iota$ is an expression of type $\iota$; the
      propositional constants $\cfalse$ and $\ctrue$ are
      expressions of type $o$.

\item If $\mathsf{f}$ is an $n$-ary function symbol and $\mathsf{E}_1,
      \ldots, \mathsf{E}_n$ are expressions of type $\iota$,
      then $(\mathsf{f}\,\,\mathsf{E}_1 \cdots \mathsf{E}_n)$ is an
      expression of type $\iota$.

\item If $\mathsf{E}_1$ is an expression of type $\rho \rightarrow
      \pi$ and $\mathsf{E}_2$ is an expression of type $\rho$, then
      $(\mathsf{E}_1\ \mathsf{E}_2)$ is an expression of type $\pi$.

\item If $\mathsf{R}$ is an argument variable of type $\rho$ and
      $\mathsf{E}$ is an expression of type $\pi$, then
      $(\lambda\mathsf{R}.\mathsf{E})$ is an expression of type
      $\rho \rightarrow \pi$.

\item If $\mathsf{E}_1,\mathsf{E}_2$ are expressions of type $\pi$,
      then  $(\mathsf{E}_1 \bigwedge_{\pi} \mathsf{E}_2)$ and
      $(\mathsf{E}_1 \bigvee_{\pi} \mathsf{E}_2)$ are expressions
      of type $\pi$.

\item If $\mathsf{E}$ is an expression of type $o$, then
      $(\mnot \mathsf{E})$ is an expression of type $o$.

\item If $\mathsf{E}_1,\mathsf{E}_2$ are expressions of type $\iota$,
      then $(\mathsf{E}_1 \approx \mathsf{E}_2)$ is an expression
      of type $o$.

\item If $\mathsf{E}$ is an expression of type $o$ and $\mathsf{R}$ is
      a variable of type $\rho$ then $(\exists_{\rho}
      \mathsf{R}\,\mathsf{E})$ is an expression of type $o$.
\end{enumerate}
\end{definition}

To denote that an expression $\mathsf{E}$ has type $\tau$ we will write
$\mathsf{E} : \tau$. The notions of \emph{free} and \emph{bound} variables
of an expression are defined as usual. An expression is called \emph{closed}
if it does not contain any free variables. An expression of type $\iota$ will
be called a {\em term}; if it does not contain any individual variables, it will
be called a {\em ground term}.
\begin{definition}
A {\em program clause} of ${\cal HOL}$ is of the form $\mathsf{p} \leftarrow_\pi \mathsf{E}$
where $\mathsf{p}$ is a predicate constant of type $\pi$ and $\mathsf{E}$
is a closed expression of type $\pi$.
A {\em program} is a finite set of program clauses.
\end{definition}
\begin{example}\label{subset-predicate-example}
We rewrite the program of Example~\ref{prolog-notation} using the syntax of ${\cal HOL}$.
For every argument type $\rho$, the {\tt subset} predicate of type $(\rho\rightarrow o)\rightarrow (\rho\rightarrow o) \rightarrow o$
takes as arguments two relations of type $\rho \rightarrow o$ and returns $\mathit{true}$ if the first relation is a subset of the second:
%
%
%
%
\begin{eqnarray*}
\mathtt{subset} \leftarrow_{(\rho\rightarrow o)\rightarrow (\rho\rightarrow o) \rightarrow o}
    \lambda\mathtt{P}\per \lambda\mathtt{Q}.
        \mnot\exists_{\rho}\mathtt{X}\lpa\lpa\mathtt{P}\ \mathtt{X}\rpa \wedge
                               \mnot\lpa\mathtt{Q}\ \mathtt{X}\rpa\rpa
\end{eqnarray*}
The use of $\lambda$-expressions obviates the need to have the formal parameters
of the predicate in the left-hand side of the definition.
\end{example}

\section{The Semantics of the Higher-Order Language ${\cal HOL}$}\label{semantics_of_language}

In this section we begin the development of the semantics of the language ${\cal HOL}$.
We start with the semantics of types, proceed with the semantics of expressions,
and then with that of programs. We assume a familiarity with the basic notions
of partially ordered sets (see~\ref{appendix-of-section-4} for the main definitions).

The semantics of the base boolean domain is three-valued. The semantics of types of
the form $\pi_1\rightarrow \pi_2$ is the set of {\em Fitting-monotonic} functions
from the domain of type $\pi_1$ to that of type $\pi_2$. We define, simultaneously
with the meaning of every type $\tau$, two partial orders on the elements of type $\tau$:
the relation $\leq_{\tau}$ which represents the {\em truth} ordering, and the relation
$\preceq_{\tau}$ which represents the {\em information} or {\em Fitting} ordering.
\begin{definition}\label{our_domains}
Let $D$ be a nonempty set. For every type $\tau$ we define recursively the set of
possible meanings of elements of ${\cal HOL}$ of type $\tau$, denoted by $\lsem \tau \rsem_D$,
as follows:
  \begin{itemize}
    \item $\lsem o \rsem _D = \{ \mfalse, \mtrue, \mundef \}$. The partial order $\leq_o$
          is the usual one induced by the ordering  $\mfalse <_o \mundef <_o \mtrue$;
          the partial order $\preceq_o$ is the one induced by the ordering  $\mundef \prec_o \mfalse$
          and $\mundef \prec_o \mtrue$.

    \item $\lsem \iota \rsem_D = D$. The partial order $\leq_\iota$ is defined as $d \leq_\iota d$ for all $d \in D$.
          The partial order $\preceq_\iota$ is also defined as $d \preceq_\iota d$ for all $d\in D$.

    \item $\lsem \iota^n \rightarrow \iota \rsem_D = D^n \rightarrow D$. No ordering relations are defined for these types.

    \item $\lsem \iota \rightarrow \pi \rsem_D = D \rightarrow \lsem \pi \rsem_D$.
          The partial order $\leq_{\iota\rightarrow\pi}$ is defined as follows: for all $f,g \in \lsem \iota \rightarrow \pi \rsem_D$,
          $f \leq_{\iota\rightarrow\pi} g$ iff $f(d) \leq_\pi g(d)$ for all $d \in D$. The partial order $\preceq_{\iota \rightarrow \pi}$
          is defined as follows: for all $f,g \in \lsem \iota \rightarrow \pi \rsem_D$, $f \preceq_{\iota\rightarrow\pi} g$
          iff $f(d) \preceq_\pi g(d)$ for all $d \in D$.

    \item $\lsem \pi_1 \rightarrow \pi_2 \rsem_D = [ \lsem \pi_1 \rsem_D \rightarrow \lsem \pi_2 \rsem_D ]$, namely the $\preceq$-monotonic
          functions\footnote{Function $f\in \lsem \pi_1 \rightarrow \pi_2\rsem_D$ is $\preceq$-monotonic if for all
          $d_1,d_2\in \lsem \pi_1\rsem_D$, $d_1\preceq_{\pi_1} d_2$ implies $f(d_1)\preceq_{\pi_2}f(d_2)$.} from \lsem $\pi_1 \rsem_D$ to $\lsem \pi_2 \rsem_D$.
          The partial order $\leq_{\pi_1 \rightarrow \pi_2}$ is defined as follows: for all $f,g \in \lsem \pi_1 \rightarrow \pi_2 \rsem_D$,
          $f \leq_{\pi_1\rightarrow\pi_2} g$ iff $f(d) \leq_{\pi_2} g(d)$ for all $d \in \lsem \pi_1 \rsem_D$. The partial order
          $\preceq_{\pi_1 \rightarrow \pi_2}$ is defined as follows: for all $f,g \in \lsem \pi_1 \rightarrow \pi_2 \rsem_D$,
          $f \preceq_{\pi_1\rightarrow\pi_2} g$ iff $f(d) \preceq_{\pi_2} g(d)$ for all $d \in \lsem \pi_1 \rsem_D$.
  \end{itemize}
\end{definition}
The subscripts in the above partial orders will often be omitted when they are obvious from context.
For every type $\pi$, the set $\lsem \pi\rsem_D$ has a least element
$\perp_{\leq_{\pi}}$ and a greatest element $\top_{\leq_{\pi}}$, called the {\em bottom} and the
{\em top} elements of $\lsem \pi\rsem_D$ with respect to $\leq_{\pi}$, respectively. In particular,
$\perp_{\leq_{o}} = \mfalse$ and $\top_{\leq_{o}} = \mtrue$; $\perp_{\leq_{\iota \rightarrow \pi}}(d) = \perp_{\leq_{\pi}}$ and
$\top_{\leq_{\iota \rightarrow \pi}}(d) = \top_{\leq_{\pi}}$, for all $d \in D$;
$\perp_{\leq_{\pi_1 \rightarrow \pi_2}}(d) = \perp_{\leq_{\pi_2}}$ and
$\top_{\leq_{\pi_1 \rightarrow \pi_2}}(d) = \top_{\leq_{\pi_2}}$, for all $d \in \lsem \pi_1 \rsem_D$.
Moreover, for every type $\pi$, the set $\lsem \pi\rsem_D$ has a least element with respect to
$\preceq_{\pi}$, denoted by $\perp_{\preceq_{\pi}}$ and called the {\em bottom}
element of $\lsem \pi\rsem_D$ with respect to $\preceq_{\pi}$. In particular,
$\perp_{\preceq_o}= \mundef$. The element $\perp_{\preceq_{\pi}}$ for $\pi \neq o$
can be defined in the obvious way as above. We will simply write $\perp$ to denote the
bottom element of any of the above partially ordered sets, when the ordering relation and the specific domain are
obvious from context.

We have the following proposition, whose proof is given in~\ref{appendix-of-section-4}:
\begin{proposition}\label{semantics_of_types_lattice_cpo}
Let $D$ be a nonempty set. For every predicate type $\pi$, $(\lsem \pi \rsem_D, \leq_\pi)$
is a complete lattice and $(\lsem \pi \rsem_D, \preceq_\pi)$ is a chain complete poset.
\end{proposition}
We can now proceed to define the semantics of ${\cal HOL}$:
\begin{definition}
A (three-valued) interpretation ${\cal I}$ of ${\cal HOL}$ consists of:
\begin{enumerate}
\item a nonempty set $D$ called the {\em domain} of ${\cal I}$;
\item an assignment to each individual constant symbol $\mathsf{c}$, of an element
      ${\cal I}(\mathsf{c}) \in D$;
\item an assignment to each predicate constant $\mathsf{p}:\pi$, of an element
      ${\cal I}(\mathsf{p}) \in \lsem \pi \rsem_D$;
\item an assignment to each function symbol $\mathsf{f} : \iota^n \to \iota$,
      of a function ${\cal I}(\mathsf{f}) \in D^n\!\rightarrow D$.
\end{enumerate}
\end{definition}
\begin{definition}
Let $D$ be a nonempty set. A {\em state} $s$ of ${\cal HOL}$ over $D$ is a function
that assigns to each argument variable $\mathsf{R}$ of type $\rho$ of ${\cal HOL}$,
an element $s(\mathsf{R}) \in \lsem \rho \rsem_D$.
\end{definition}

We define: $\mathit{true}^{-1}=\mathit{false}$, $\mathit{false}^{-1}=\mathit{true}$
and $\mathit{undef}^{-1}=\mathit{undef}$.
\begin{definition}\label{definition-semantics-of-expressions}
Let $D$ be a nonempty set, let ${\cal I}$ be an interpretation over $D$,
and let $s$ be a state over $D$. The semantics of expressions of ${\cal HOL}$
with respect to ${\cal I}$ and $s$, is defined as follows:
\begin{enumerate}
\item $\lsem \cfalse \rsem_s ({\cal I}) = \mathit{false}$, and $\lsem \ctrue \rsem_s ({\cal I}) = \mathit{true}$


\item $\lsem \mathsf{c} \rsem_s ({\cal I}) = {\cal I}(\mathsf{c})$, for every
      individual constant $\mathsf{c}$

\item $\lsem \mathsf{p} \rsem_s ({\cal I}) = {\cal I}(\mathsf{p})$, for every
      predicate constant $\mathsf{p}$

\item $\lsem \mathsf{R} \rsem_s ({\cal I}) = s(\mathsf{R})$, for every
      argument variable $\mathsf{R}$

\item $\lsem (\mathsf{f}\,\,\mathsf{E}_1\cdots \mathsf{E}_n) \rsem_s ({\cal I}) =
      {\cal I}(\mathsf{f})\,\,\lsem \mathsf{E}_1\rsem_s ({\cal I}) \cdots \lsem \mathsf{E}_n\rsem_s ({\cal I})$,
      for every $n$-ary function symbol $\mathsf{f}$

\item $\lsem \mathsf{(}\mathsf{E}_1\mathsf{E}_2\mathsf{)} \rsem_s
      ({\cal I})= \lsem \mathsf{E}_1 \rsem_s ({\cal I})(\lsem \mathsf{E}_2\rsem_s({\cal I}))$

\item $\lsem \mathsf{(\lambda R.E)} \rsem_s ({\cal I}) =\lambda d.\lsem
      \mathsf{E}\rsem_{s[\mathsf{R}/d]}({\cal I})$, where if $\mathsf{R}:\rho$ then $d$ ranges over
      $\lsem \rho \rsem_D$

\item $\lsem (\mathsf{E}_1 \bigvee_{\pi} \mathsf{E}_2)\rsem_s ({\cal I}) =
      \bigvee_{\leq_{\pi}}\{\lsem \mathsf{E}_1\rsem_s({\cal I}),\lsem \mathsf{E}_2\rsem_s({\cal I})\}$

\item $\lsem (\mathsf{E}_1 \bigwedge_{\pi} \mathsf{E}_2)\rsem_s ({\cal I}) =
      \bigwedge_{\leq_{\pi}}\{\lsem \mathsf{E}_1\rsem_s({\cal I}),\lsem \mathsf{E}_2\rsem_s({\cal I})\}$

\item $\lsem (\mnot \mathsf{E}) \rsem_s ({\cal I}) = (\lsem \mathsf{E} \rsem_s ({\cal I}))^{-1}$

\item $\lsem (\mathsf{E}_1 \,\mathsf{\approx}\, \mathsf{E}_2)\rsem_s ({\cal I}) = \left\{\begin{array}{ll}
                                               \mathit{true}, & \mbox{if $\lsem \mathsf{E}_1 \rsem_s({\cal I}) = \lsem \mathsf{E}_2 \rsem_s({\cal I})$}\\
                                               \mathit{false}, & \mbox{otherwise}
                                                   \end{array} \right. $

\item $\lsem (\exists_{\rho} \mathsf{R}\, \mathsf{E}) \rsem_s ({\cal I})= \bigvee_{\leq_o} \{\lsem \mathsf{E} \rsem_{s[\mathsf{R}/d]}({\cal I}) \mid d \in \lsem \rho\rsem_D\}$
\end{enumerate}
\end{definition}

For closed expressions $\mathsf{E}$ we will often write $\lsem
\mathsf{E} \rsem({\cal I})$ instead of $\lsem \mathsf{E} \rsem_s({\cal I})$
(since, in this case, the meaning of $\mathsf{E}$ is independent of
$s$). The following lemma demonstrates that our semantic valuation
function returns elements that belong to the appropriate domain
(the proof of the lemma by structural induction on $\mathsf{E}$, is easy
and omitted).
\begin{lemma}\label{expressions-well-defined}
Let $\mathsf{E} : \rho$ be an expression and let $D$ be a nonempty set.
Moreover, let $s$ be a state over $D$ and let ${\cal I}$ be an interpretation over $D$.
Then, $\lsem \mathsf{E} \rsem_s({\cal I}) \in \lsem \rho \rsem_D$.
\end{lemma}

Finally, we define the notion of {\em model} for ${\cal HOL}$ programs:
\begin{definition}
Let $\mathsf{P}$ be a ${\cal HOL}$ program and let $M$ be an interpretation of $\mathsf{P}$.
Then $M$ will be called a {\em model} of $\mathsf{P}$ iff for all clauses
$\mathsf{p} \leftarrow_\pi \mathsf{E}$ of $\mathsf{P}$,
it holds $\lsem \mathsf{E} \rsem(M) \leq_\pi M(\mathsf{p})$.
\end{definition}

\section{An Alternative View of Fitting-Monotonic Functions}\label{bijection}
In this section we demonstrate that every Fitting-monotonic function $f$
can be equivalently represented as a pair of functions $(f_1,f_2)$, where $f_1$ is
monotone-antimonotone, $f_2$ is antimonotone-monotone and $f_1\leq f_2$.
Consider for example a function  $f$ of type $o \rightarrow o$, i.e.,
$f:\{\textit{true},\textit{false},\textit{undef}\} \rightarrow \{\textit{true},\textit{false},\textit{undef}\}$.
One can view the truth values as pairs where $\textit{true}$ corresponds to
$(\textit{true},\textit{true})$, $\textit{false}$ corresponds to
$(\textit{false},\textit{false})$, and $\textit{undef}$ corresponds to
$(\textit{false},\textit{true})$. Therefore, $f$ can also equivalently
be seen as a function $f'$ that takes pairs and returns pairs. We can then
``break'' $f'$ into two components $f_1$ and $f_2$ where $f_1$ returns the
first element of the pair that $f'$ returns while $f_2$ returns the second.
The monotone-antimonotone and antimonotone-monotone requirements ensure that
the pair $(f_1,f_2)$ retains the property of Fitting-monotonicity of the
original function $f$. These ideas can be generalized to arbitrary types.
The formal details of this equivalence are described below. The following definitions
will be used:
\begin{definition}
Let $L_1,L_2$ be sets and let $\leq$ be a partial order on $L_1\cup L_2$.
We define: $L_1 \otimes_{\leq} L_2 = \{(x,y)\in L_1\times L_2: x \leq y\}$.
\end{definition}
We will omit the $\leq$ from $\otimes_{\leq}$ when it is obvious from context.
\begin{definition}
Let $L_1,L_2$ be sets and let $\leq$ be a partial order on $L_1\cup L_2$.
Also, let $(A,\leq_A)$ be a partially ordered set. A function
$f:(L_1\otimes L_2) \rightarrow A$ will be called
{\em monotone-antimonotone} (respectively {\em antimonotone-monotone}) if for all
$(x,y),(x',y') \in L_1\otimes L_2$ with $x \leq x'$ and $y' \leq y$, it holds that
$f(x,y) \leq_A f(x',y')$ (respectively $f(x',y') \leq_A f(x,y)$).
We denote by $[(L_1\otimes L_2) \stackrel{\mathsf{ma}}\rightarrow A]$ the
set of functions that are monotone-antimonotone and by
$[(L_1\otimes L_2) \stackrel{\mathsf{am}}\rightarrow A]$
those that are antimonotone-monotone.
\end{definition}

In order to establish the bijection between Fitting-monotonic functions
and pairs of monotone-antimonotone and antimonotone-monotone functions,
we reinterpret the predicate types of ${\cal HOL}$ in an alternative way.
\begin{definition}\label{alternative_denotation_of_types}
Let $D$ be a nonempty set. For every type $\tau$ we define the monotone-antimonotone
and the antimonotone-monotone meanings of the elements of type $\tau$
with respect to $D$, denoted respectively by $\lsem \tau \rsem_{D}^\mathsf{ma}$
and $\lsem \tau \rsem_D^\mathsf{am}$.
At the same time we define a partial order $\leq_\tau$ between the elements
of $\lsem \tau \rsem_D^\mathsf{ma} \cup \lsem \tau \rsem_D^\mathsf{am}$.
\begin{itemize}
  \item $\lsem o \rsem_D^\mathsf{ma} = \lsem o \rsem_D^\mathsf{am} = \{ \mfalse, \mtrue \}$.
       The partial order $\leq_o$ is the usual one induced by the ordering $\mfalse \leq_o \mtrue$.

  \item $\lsem \iota \rsem^\mathsf{ma} = \lsem \iota \rsem^\mathsf{am} = D$.
       The partial order $\leq_\iota$ is defined as $d \leq_\iota d$,
       for all $d\in D$.

  \item $\lsem \iota^n \rightarrow \iota \rsem^\mathsf{ma} = \lsem \iota^n \rightarrow \iota \rsem^\mathsf{am} = D^n \rightarrow D$. There is no partial order for elements of type $\iota^n \rightarrow \iota$.

  \item $\lsem \iota \rightarrow \pi \rsem_D^\mathsf{ma} = D\rightarrow \lsem \pi \rsem_D^\mathsf{ma}$
        and $\lsem \iota \rightarrow \pi \rsem_D^\mathsf{am} = D \rightarrow \lsem \pi \rsem_D^\mathsf{am}$.
        The partial order $\leq_{\iota \rightarrow \pi}$ is defined as follows:
        for all $f,g \in \lsem \iota \rightarrow \pi \rsem_D^\mathsf{ma} \cup \lsem \iota \rightarrow \pi \rsem_D^\mathsf{am}$,
        $f \leq_{\iota\rightarrow\pi} g$ iff $f(d) \leq_\pi g(d)$ for all $d \in D$.

  \item $\lsem \pi_1 \rightarrow \pi_2 \rsem_D^\mathsf{ma} =
        [ (\lsem \pi_1 \rsem_D^\mathsf{ma} \otimes \lsem \pi_1 \rsem_D^\mathsf{am}) \stackrel{\mathsf{ma}}{\rightarrow} \lsem \pi_2 \rsem_D^\mathsf{ma}]$,
        and $\lsem \pi_1 \rightarrow \pi_2 \rsem_D^\mathsf{am} =
        [ (\lsem \pi_1 \rsem_D^\mathsf{ma} \otimes \lsem \pi_1 \rsem_D^\mathsf{am}) \stackrel{\mathsf{am}}\rightarrow \lsem \pi_2 \rsem_D^\mathsf{am}]$.
        The relation $\leq_{\pi_1 \rightarrow \pi_2}$ is the partial order defined as follows:
        for all $f,g \in \lsem \pi_1 \rightarrow \pi_2 \rsem_D^\mathsf{ma} \cup \lsem \pi_1 \rightarrow \pi_2 \rsem_D^\mathsf{am}$,
        $f \leq_{\pi_1\rightarrow\pi_2} g$ iff $f(d_1,d_2) \leq_{\pi_2} g(d_1,d_2)$ for all
        $(d_1,d_2) \in \lsem \pi_1 \rsem_D^\mathsf{ma} \otimes \lsem \pi_1 \rsem_D^\mathsf{am}$.

\end{itemize}
\end{definition}

For every $\pi$, the bottom and top elements of $\lsem \pi \rsem_D^\mathsf{ma}$ and
$\lsem \pi \rsem_D^\mathsf{am}$ can be defined in the obvious way. We have the following
proposition (see~\ref{appendix-of-section-5} for the proof):
\begin{proposition}\label{ma-am-lattices}
Let $D$ be a nonempty set. For every predicate type $\pi$, $(\lsem \pi \rsem_D^{\mathsf{ma}}, \leq_\pi)$ and
$(\lsem \pi \rsem_D^{\mathsf{am}}, \leq_\pi)$ are complete lattices.
\end{proposition}

We extend, in a pointwise way, our orderings to apply to pairs. For simplicity, we overload
our notation and use the same symbols $\leq$ and $\preceq$ for the new orderings.
\begin{definition}\label{orderings_on_pairs}
Let $D$ be a nonempty set and let $\pi$ be a predicate type. We define the relations $\leq_{\pi}$ and  $\preceq_{\pi}$, so that
for all $(x,y),(x',y') \in \lsem \pi \rsem_D^{\mathsf{ma}} \otimes \lsem \pi \rsem_D^{\mathsf{am}}$:
\begin{itemize}
\item $(x, y) \leq_{\pi} (x', y')$ iff $x \leq_{\pi} x'$ and $y \leq_{\pi} y'$.
\item $(x, y) \preceq_{\pi} (x', y')$ iff $x \leq_{\pi} x'$ and $y' \leq_{\pi} y$.
\end{itemize}
\end{definition}
The following proposition is demonstrated in~\ref{appendix-of-section-5}:
\begin{proposition}\label{pairs_complete_lattice_cpo}
Let $D$ be a nonempty set. For each predicate type $\pi$, $\lsem \pi \rsem_D^{\mathsf{ma}} \otimes \lsem \pi \rsem_D^{\mathsf{am}}$
is a complete lattice with respect to $\leq_{\pi}$ and a chain-complete poset with respect to $\preceq_{\pi}$.
\end{proposition}

In the rest of the paper we will denote the {\em first} and {\em second} selection
functions on pairs with the more compact notation $[\cdot]_1$ and $[\cdot]_2$:
given any pair $(x,y)$, it is $[(x,y)]_1 = x$ and $[(x,y)]_2 = y$.
We can now establish the bijection between $\lsem \pi \rsem_D$ and
$\lsem \pi \rsem_D^{\mathsf{ma}} \otimes \lsem \pi \rsem_D^{\mathsf{am}}$.
The following definition and two propositions (whose proofs are given in~\ref{appendix-of-section-5}), explain how.
\begin{definition}
Let $D$ be a nonempty set. For every predicate type $\pi$, we define recursively the functions
$\tau_\pi: \lsem \pi \rsem_D \rightarrow (\lsem \pi \rsem_D^{\mathsf{ma}} \otimes \lsem \pi \rsem_D^{\mathsf{am}})$ and
$\tau^{-1}_\pi: (\lsem \pi \rsem_D^{\mathsf{ma}} \otimes \lsem \pi \rsem_D^{\mathsf{am}}) \rightarrow \lsem \pi \rsem_D$,
as follows.
\begin{itemize}
  \item $\tau_o(\mfalse) = (\mfalse, \mfalse)$, $\tau_o(\mtrue) = (\mtrue, \mtrue)$,
                        $\tau_o(\mundef) = (\mfalse, \mtrue)$

  \item $\tau_{\iota\rightarrow\pi}(f) = (\lambda d. [\tau_\pi(f(d))]_1, \lambda d. [\tau_\pi(f(d))]_2)$

  \item $\tau_{\pi_1\rightarrow\pi_2}(f) = (\lambda(d_1, d_2). [\tau_{\pi_2}(f(\tau^{-1}_{\pi_1}(d_1, d_2)))]_1, \lambda(d_1, d_2). [\tau_{\pi_2}(f(\tau^{-1}_{\pi_1}(d_1, d_2)))]_2)$

\end{itemize}
and
\begin{itemize}
\item $\tau^{-1}_o(\mfalse,\mfalse) = \mfalse$, $\tau^{-1}_o(\mtrue,\mtrue) = \mtrue$,
                              $\tau^{-1}_o(\mfalse, \mtrue) = \mundef$

\item $\tau^{-1}_{\iota\rightarrow\pi}(f_1, f_2) = \lambda d. \tau_\pi^{-1}(f_1(d), f_2(d))$

\item $\tau^{-1}_{\pi_1\rightarrow\pi_2}(f_1, f_2) = \lambda d. \tau^{-1}_{\pi_2}(f_1(\tau_{\pi_1}(d)), f_2(\tau_{\pi_1}(d)))$.

\end{itemize}
\end{definition}
\begin{proposition}\label{tau-monotonic}
Let $D$ be a nonempty set and let $\pi$ be a predicate type. Then, for every $f,g \in \lsem \pi \rsem_D$
and for every $(f_1,f_2),(g_1,g_2) \in \lsem \pi \rsem_D^{\mathsf{ma}}\otimes \lsem \pi \rsem_D^{\mathsf{am}}$,
the following statements hold:
  \begin{enumerate}
  \item $\tau_{\pi}(f) \in  (\lsem \pi \rsem_D^{\mathsf{ma}} \otimes \lsem \pi \rsem_D^{\mathsf{am}})$
        and $\tau^{-1}_{\pi}(f_1,f_2) \in \lsem \pi \rsem_D$.
  \item If $f \preceq_{\pi} g$ then $\tau_\pi(f) \preceq_{\pi} \tau_\pi(g)$.
  \item If $f \leq_{\pi} g$ then $\tau_\pi(f) \leq_{\pi} \tau_\pi(g)$.
  \item If $(f_1,f_2) \preceq_{\pi} (g_1,g_2)$ then $\tau_\pi^{-1}(f_1,f_2) \preceq_{\pi} \tau_\pi^{-1}(g_1,g_2)$.
  \item If $(f_1,f_2) \leq_{\pi} (g_1,g_2)$ then $\tau^{-1}_\pi(f_1,f_2) \leq_{\pi} \tau^{-1}_\pi(g_1,g_2)$.
\end{enumerate}
\end{proposition}
\begin{proposition}\label{tau-elim}
Let $D$ be a nonempty set and let $\pi$ be a predicate type. Then, for every $f \in \lsem \pi \rsem_D$, $\tau^{-1}_{\pi}(\tau_{\pi}(f))=f$,
and for every $(f_1,f_2) \in \lsem \pi \rsem_D^{\mathsf{ma}}\otimes \lsem \pi \rsem_D^{\mathsf{am}}$,
$\tau_{\pi}(\tau^{-1}_{\pi}(f_1,f_2))=(f_1,f_2)$.
\end{proposition}

\section{The Well-Founded Semantics for ${\cal HOL}$ Programs}\label{well_founded}
In this section we demonstrate that every program of ${\cal HOL}$ has a distinguished
{\em minimal Herbrand model} which can be obtained by an iterative procedure. This construction
generalizes the familiar well-founded semantics. Our main
results are based on a mild generalization of the consistent approximation fixpoint
theory of~\cite{DMT04}. We start with the relevant definitions.
\begin{definition}
Let $\mathsf{P}$ be a program. The Herbrand universe $U_{\mathsf{P}}$ of $\mathsf{P}$ is the set of all ground terms
that can be formed out of the individual constants\footnote{As usual, if $\mathsf{P}$ has no constants, we assume the existence of an arbitrary one.}
and the function symbols of $\mathsf{P}$.
\end{definition}

\begin{definition}
A (three-valued) {\em Herbrand interpretation} ${\cal I}$ of a program $\mathsf{P}$ is an interpretation such that:
\begin{enumerate}
\item the domain of ${\cal I}$ is the Herbrand universe $U_{\mathsf{P}}$ of $\mathsf{P}$;
\item for every individual constant $\mathsf{c}$ of $\mathsf{P}$, ${\cal I}(\mathsf{c})=\mathsf{c}$;
\item for every predicate constant $\mathsf{p}:\pi$ of $\mathsf{P}$, ${\cal I}(\mathsf{p}) \in \lsem \pi \rsem_{U_{\mathsf{P}}}$;
\item for every $n$-ary function symbol $\mathsf{f}$ of $\mathsf{P}$ and for all $\mathsf{t}_1,\ldots,\mathsf{t}_n \in U_{\mathsf{P}}$, ${\cal I}(\mathsf{f})\, \mathsf{t}_1\cdots \mathsf{t}_n = \mathsf{f}\,\mathsf{t}_1\cdots\mathsf{t}_n$.
\end{enumerate}
\end{definition}
We denote the set of all three-valued Herbrand interpretations of a program $\mathsf{P}$ by ${\cal H}_\mathsf{P}$.
A {\em Herbrand state} of $\mathsf{P}$ is a state whose underlying domain is $U_{\mathsf{P}}$.
A {\em Herbrand model} of $\mathsf{P}$ is a Herbrand interpretation that is a model of $\mathsf{P}$.
The truth and the information orderings easily extend to Herbrand interpretations:
\begin{definition}\label{ordering-interpretations}
Let $\mathsf{P}$ be a program. We define the partial orders $\leq$ and $\preceq$ on ${\cal H}_\mathsf{P}$
as follows: for all ${\cal I}, {\cal J} \in  {\cal H}_\mathsf{P}$, ${\cal I} \leq {\cal J}$ (respectively,
${\cal I} \preceq {\cal J}$) iff for every predicate type $\pi$ and for every predicate constant
$\mathsf{p} : \pi$ of $\mathsf{P}$, ${\cal I}(\mathsf{p}) \leq_\pi {\cal J}(\mathsf{p})$ (respectively,
${\cal I}(\mathsf{p}) \preceq_\pi {\cal J}(\mathsf{p})$).
\end{definition}
The proof of the following proposition is analogous to that of Proposition~\ref{semantics_of_types_lattice_cpo} and omitted:
\begin{proposition}\label{semantics_of_interpretations_lattice_cpo}
Let $\mathsf{P}$ be a program. Then, $({\cal H}_{\mathsf{P}}, \leq)$ is a complete lattice and
$({\cal H}_{\mathsf{P}}, \preceq)$ is a chain complete poset.
\end{proposition}
The following lemma is also easy to establish, and its proof is omitted:
\begin{lemma}\label{expressions_fitting}
Let $\mathsf{P}$ be a program, let ${\cal I},{\cal J}\in {\cal H}_\mathsf{P}$, and let $s$ be a Herbrand state of $\mathsf{P}$.
For every expression $\mathsf{E}$, if ${\cal I} \preceq {\cal J}$ then $\lsem \mathsf{E} \rsem_s({\cal I}) \preceq \lsem \mathsf{E} \rsem_s({\cal J})$.
\end{lemma}

The bijection established in Section~\ref{bijection} extends also to interpretations.
More specifically, every three-valued Herbrand interpretation ${\cal I}$ of a program
$\mathsf{P}$ can be mapped by (an extension of) $\tau$ to a pair of functions $(I,J)$ such that:
\begin{itemize}
\item for every individual constant $\mathsf{c}$ of $\mathsf{P}$, $I(\mathsf{c})=J(\mathsf{c})=\mathsf{c}$;
\item for every predicate constant $\mathsf{p}:\pi$ of $\mathsf{P}$, $I(\mathsf{p}) \in \lsem \pi \rsem^{\mathsf{ma}}_{U_{\mathsf{P}}}$
      and $J(\mathsf{p}) \in \lsem \pi \rsem^{\mathsf{am}}_{U_{\mathsf{P}}}$;
\item for every $n$-ary function symbol $\mathsf{f}$ of $\mathsf{P}$ and for all $\mathsf{t}_1,\ldots,\mathsf{t}_n \in U_{\mathsf{P}}$, $I(\mathsf{f})\, \mathsf{t}_1\cdots \mathsf{t}_n = J(\mathsf{f})\, \mathsf{t}_1\cdots \mathsf{t}_n =\mathsf{f}\,\mathsf{t}_1\cdots\mathsf{t}_n$.
\end{itemize}

Functions of the form $I$ above will be called ``{\em monotone-antimonotone Herbrand interpretations}''
and functions of the form $J$ will be called  ``{\em antimonotone-monotone Herbrand interpretations}''.
We will denote by ${\cal H}^{\mathsf{ma}}_\mathsf{P}$ the set of functions of the former type
and by ${\cal H}^{\mathsf{am}}_\mathsf{P}$ those of the latter type. As in Definition~\ref{ordering-interpretations},
we can define a partial order $\leq$ on ${\cal H}^{\mathsf{ma}}_\mathsf{P} \cup {\cal H}^{\mathsf{am}}_\mathsf{P}$.
Similarly, as in Definition~\ref{orderings_on_pairs}, we can define partial orders
$\leq$ and $\preceq$ on ${\cal H}^{\mathsf{ma}}_\mathsf{P} \otimes {\cal H}^{\mathsf{am}}_\mathsf{P}$.
The proof of the following proposition is a direct consequence of the proofs of Propositions~\ref{pairs_complete_lattice_cpo}
and~\ref{ma-am-lattices}, and therefore omitted.
\begin{proposition}\label{interpretation_pairs_complete_lattice_cpo}
Let $\mathsf{P}$ be a program.  Then, $({\cal H}^{\mathsf{ma}}_\mathsf{P}, \leq)$  and $({\cal H}^{\mathsf{am}}_\mathsf{P},\leq)$ are complete
lattices having the same $\perp$ and $\top$ elements. Moreover, $({\cal H}^{\mathsf{ma}}_\mathsf{P} \otimes {\cal H}^{\mathsf{am}}_\mathsf{P}, \leq)$
is a complete lattice and $({\cal H}^{\mathsf{ma}}_\mathsf{P} \otimes {\cal H}^{\mathsf{am}}_\mathsf{P}, \preceq)$
is a chain-complete poset.
\end{proposition}

The bijection between ${\cal H}_{\mathsf{P}}$ and ${\cal H}^{\mathsf{ma}}_\mathsf{P} \otimes {\cal H}^{\mathsf{am}}_\mathsf{P}$
can be explained more formally as follows. Given ${\cal I} \in {\cal H}_{\mathsf{P}}$, we define $\tau({\cal I}) = (I,J)$, where
for every predicate constant $\mathsf{p}:\pi$ it holds $I(\mathsf{p}) = [\tau_{\pi}({\cal I}(\mathsf{p}))]_1$ and $J(\mathsf{p}) = [\tau_{\pi}({\cal I}(\mathsf{p}))]_2$.
Conversely, given a pair $(I,J) \in {\cal H}^{\mathsf{ma}}_\mathsf{P} \otimes {\cal H}^{\mathsf{am}}_\mathsf{P}$,
we define the three-valued Herbrand interpretation ${\cal I}$ as follows:
${\cal I}(\mathsf{p}) = \tau^{-1}_{\pi}(I(\mathsf{p}),J(\mathsf{p}))$. We now define
the three-valued and two-valued immediate consequence operators:
\begin{definition}
Let $\mathsf{P}$ be a program. The {\em three-valued immediate consequence operator}
$\Psi_\mathsf{P}: {\cal H}_{\mathsf{P}} \rightarrow {\cal H}_{\mathsf{P}}$ of $\mathsf{P}$ is defined for every $\mathsf{p}:\pi$ as:
$\Psi_\mathsf{P}({\cal I})(\mathsf{p}) = \bigvee_{\leq_{\pi}}\{ \lsem \mathsf{E} \rsem({\cal I}) \mid (\mathsf{p} \leftarrow_{\pi} \mathsf{E}) \in \mathsf{P}\}$.
\end{definition}
\begin{definition}
Let $\mathsf{P}$ be a program. The {\em two-valued immediate consequence operator}
$T_\mathsf{P}:({\cal H}^{\mathsf{ma}}_\mathsf{P} \otimes  {\cal H}^{\mathsf{am}}_\mathsf{P}) \rightarrow ({\cal H}^{\mathsf{ma}}_\mathsf{P} \otimes  {\cal H}^{\mathsf{am}}_\mathsf{P})$ of $\mathsf{P}$ is defined as: $T_\mathsf{P}(I,J) = \tau(\Psi_\mathsf{P}(\tau^{-1}(I,J)))$.
\end{definition}
From Proposition~\ref{tau-monotonic-for-interpretations} in~\ref{appendix-of-section-6b} 
it follows that $T_{\mathsf{P}}$ is well-defined. Moreover, it is Fitting-monotonic as the following lemma
demonstrates (see~\ref{appendix-of-section-6b} for the proof):
\begin{lemma}\label{T_P_is_Fitting_monotonic}
Let $\mathsf{P}$ be a program and let  $(I_1, J_1), (I_2, J_2) \in {\cal H}^{\mathsf{ma}}_\mathsf{P} \otimes  {\cal H}^{\mathsf{am}}_\mathsf{P}$.
If $(I_1, J_1) \preceq (I_2, J_2)$ then $T_\mathsf{P}(I_1, J_1) \preceq T_\mathsf{P}(I_2,J_2)$.
\end{lemma}

We will use $T_{\mathsf{P}}$ to construct the well-founded model of program $\mathsf{P}$.
Our construction is based on a mild extension of consistent approximation fixpoint theory~\cite{DMT04}.
Therefore, in order for the following two definitions and subsequent theorem to be fully comprehended,
it would be helpful if the reader had some familiarity with the material in~\cite{DMT04}.
\begin{definition}
Let $\mathsf{P}$ be a program and let $(I,J) \in {\cal H}^{\mathsf{ma}}_\mathsf{P} \otimes  {\cal H}^{\mathsf{am}}_\mathsf{P}$.
Assume that $(I,J) \preceq T_{\mathsf{P}}(I,J)$. We define $I^{\uparrow} = \textit{lfp}([T_{\mathsf{P}}(I,\cdot)]_2)$
and $J^{\downarrow} = \textit{lfp}([T_{\mathsf{P}}(\cdot,J)]_1)$, where by $T_{\mathsf{P}}(\cdot,J)$ we denote the
function $f(x) = T_{\mathsf{P}}(x,J)$ and by $T_{\mathsf{P}}(I,\cdot)$ the function $g(x) = T_{\mathsf{P}}(I,x)$.
\end{definition}

It can be shown (see~\ref{appendix-of-section-6a}) that $I^{\uparrow}$ and $J^{\downarrow}$
are well-defined, and this is due to the crucial assumption $(I,J) \preceq T_{\mathsf{P}}(I,J)$. This property was
introduced in~\cite{DMT04} where it is named $A$-reliability (in our case $A$ is the
$T_{\mathsf{P}}$ operator). Before proceeding to the definition of the well-founded
semantics, we need to define one more operator, namely the {\em stable revision operator}
(see~\cite{DMT04}[page 91] for the intuition and motivation behind this operator).
\begin{definition}
Let $\mathsf{P}$ be a program. We define the function ${\cal C}_{T_{\mathsf{P}}}$ which for every pair
$(I,J) \in {\cal H}^{\mathsf{ma}}_\mathsf{P} \otimes  {\cal H}^{\mathsf{am}}_\mathsf{P}$ with $(I,J) \preceq T_{\mathsf{P}}(I,J)$,
returns the pair $(J^{\downarrow},I^{\uparrow})$:
$${\cal C}_{T_{\mathsf{P}}}(I,J) = (J^{\downarrow},I^{\uparrow}) = (\textit{lfp}([T_{\mathsf{P}}(\cdot,J)]_1), \textit{lfp}([T_{\mathsf{P}}(I,\cdot)]_2))$$
The function ${\cal C}_{T_{\mathsf{P}}}$ will be called the {\em stable revision operator} for
$T_{\mathsf{P}}$.
\end{definition}
The following theorem is a direct consequence of Theorem~\ref{central-theorem} given in~\ref{appendix-of-section-6a} 
(which extends Theorem 3.11 in~\cite{DMT04} to our case):
\begin{theorem}\label{iterative_definition_of_wfm}
Let $\mathsf{P}$ be a program. We define the following sequence of pairs of interpretations:
\[
\begin{array}{llll}
(I_0,J_0) & = & (\perp,\top) & \\
(I_{\lambda+1},J_{\lambda+1}) & = & {\cal C}_{T_{\mathsf{P}}}(I_{\lambda},J_{\lambda}) &\\
(I_\lambda, J_\lambda) & = & \bigvee_{\preceq}\{(I_{\kappa},J_{\kappa}) \mid \kappa < \lambda\} & \mbox{ for limit ordinals $\lambda$}
\end{array}
\]
Then, the above sequence of pairs of interpretations is well-defined. Moreover, there exists a
least ordinal $\delta$ such that $(I_{\delta},J_{\delta}) = {\cal C}_{T_{\mathsf{P}}}(I_\delta,J_\delta)$
and $(I_{\delta},J_{\delta}) \in {\cal H}^{\mathsf{ma}}_\mathsf{P} \otimes  {\cal H}^{\mathsf{am}}_\mathsf{P}$.
\end{theorem}

In the following, we will denote with ${\cal M}_{\mathsf{P}}$ the interpretation $\tau^{-1}(I_{\delta},J_{\delta})$.
The following two lemmas demonstrate that the pre-fixpoints of $T_{\mathsf{P}}$ correspond exactly
to the three-valued models of $\mathsf{P}$ (see~\ref{appendix-of-section-6b} for the corresponding proofs).
\begin{lemma}\label{tp-fixpoint-model}
Let $\mathsf{P}$ be a program. If $(I, J)\in {\cal H}^{\mathsf{ma}}_\mathsf{P} \otimes  {\cal H}^{\mathsf{am}}_\mathsf{P}$
is a pre-fixpoint of $T_\mathsf{P}$ then $\tau^{-1}(I, J)$ is a model of $\mathsf{P}$.
\end{lemma}


%
\begin{lemma}\label{is_a_prefixpoint}
Let ${\cal M}\in {\cal H}_{\mathsf{P}}$ be a model of $\mathsf{P}$. Then, $\tau({\cal M})$ is a pre-fixpoint of $T_\mathsf{P}$.
\end{lemma}

Finally, the following two lemmas (see~\ref{appendix-of-section-6b} for the proofs), provide
evidence that ${\cal M}_{\mathsf{P}}$ is an extension of the classical well-founded semantics
to the higher-order case:
\begin{theorem}\label{minimal-model}
Let $\mathsf{P}$ be a program. Then, ${\cal M}_{\mathsf{P}}$ is a $\leq$-minimal model of $\mathsf{P}$.
\end{theorem}
\begin{theorem}\label{backwards-compatible}
For every propositional program $\mathsf{P}$, ${\cal M}_{\mathsf{P}}$ coincides with
the well-founded model of $\mathsf{P}$.
\end{theorem}
In~\ref{appendix-of-section-6c} we give an example construction of ${\cal M}_{\mathsf{P}}$ for a
given program $\mathsf{P}$.



\section{Related and Future Work}\label{conclusions}
In this section we compare our technique with the existing proposals
for assigning semantics to higher-order logic programs with negation
and we discuss possibly fruitful directions for future research.

The proposed extensional three-valued approach has important differences from the
existing alternative ones, namely~\cite{CharalambidisER14},~\cite{RS16} and~\cite{RondogiannisS17}.
As already mentioned in the introduction section, the technique in~\cite{RondogiannisS17}
is not extensional in the general case (it is however extensional if the source higher-order
programs are {\em stratified} - see~\cite{RondogiannisS17} for the formal definition of this notion).
In this respect, the present approach is more general since it assigns an extensional semantics to {\em all}
the programs of ${\cal HOL}$.

On the other hand, both of the techniques~\cite{CharalambidisER14} and~\cite{RS16} rely on an infinite-valued
logic, and give a very fine-grained semantics to programs. This fine-grained nature of the
infinite-valued approach makes it very appealing from a mathematical point of view.
As it was recently demonstrated in~\cite{Esik15,CarayolE16}, in the case of first-order logic
programs the infinite-valued approach satisfies all identities of {\em iteration theories}~\cite{BloomE93},
while the well-founded semantics does not. Since iteration theories provide an
abstract framework for the evaluation of the merits of various semantic approaches for
languages that involve recursion, these results appear to suggest that the infinite-valued
approach has advantages from a mathematical point of view. On the other hand, the well-founded
semantics is based on a much simpler three-valued logic, it is widely known
to the logic programming community, and it has been studied and used for almost
three decades. It is important however to emphasize that the differences between
the infinite-valued and the well-founded approaches are not only a matter of
mathematical elegance. In many programs, the two techniques behave differently.
For example, given the program:
\[
\begin{array}{l}
\mbox{\tt p $\leftarrow$ $\mysim$ ($\mysim$ p)}
\end{array}
\]
the approaches in~\cite{CharalambidisER14} and~\cite{RS16} will produce the
model $\{({\tt p},\textit{undef})\}$, while our present approach will produce the model
$\{({\tt p},\textit{false})\}$. In essence, our present approach
{\em cancels such nested negations} (see also the discussion in~\cite{DBV12}[page 185, Example 1] on this issue),
while the approaches in~\cite{CharalambidisER14} and~\cite{RS16} assign the value $\textit{undef}$
due to the circular dependence of {\tt p} on itself through negation.

Similarly, for the following program (taken from~\cite{RondogiannisS17}):
\[
\begin{array}{l}
\mbox{\tt s  $\leftarrow$ $\lambda$Q.Q (s Q)}\\
\mbox{\tt p  $\leftarrow$ $\lambda$R.R}\\
\mbox{\tt q  $\leftarrow$ $\lambda$R.$\mysim$ (w R)}\\
\mbox{\tt w  $\leftarrow$ $\lambda$R.($\mysim$ R)}
\end{array}
\]
the infinite-valued approaches will return the value $\textit{false}$ for the query
{\tt (s p)} and $\textit{undef}$ for {\tt (s q)}, while our present approach will
return the value $\textit{false}$ for both queries.

It is an interesting topic for future research to identify large classes of programs
where the infinite-valued approach and the present one coincide. Possibly a good candidate
for such a comparison would be the class of stratified higher-order logic programs~\cite{RS16}.
More generally, we believe that an investigation of the connections between the
well-founded semantics and the infinite-valued one, will be quite rewarding.

Another interesting direction for future research would be to consider other
possible semantics that can be revealed using approximation fixpoint theory.
It is well-known that for first-order logic programs, approximation fixpoint theory
can be used in order to define other useful fixpoints such as {\em stable},
{\em Kripke-Kleene}, and {\em supported} ones. We argue that using the approach
proposed in this paper, this can also be done for higher-order logic programs.
In particular, as in the first-order case, the fixpoints of $T_{\mathsf{P}}$
correspond to {\em $3$-valued supported models of} $\mathsf{P}$ (recall that by Lemma~\ref{tp-fixpoint-model}
every fixpoint of $T_{\mathsf{P}}$ is a model of $\mathsf{P}$). Moreover,
since $T_{\mathsf{P}}$ is Fitting-monotonic over
${\cal H}^{\mathsf{ma}}_\mathsf{P} \otimes  {\cal H}^{\mathsf{am}}_\mathsf{P}$
(which by Proposition~\ref{interpretation_pairs_complete_lattice_cpo} is a
chain-complete poset), it has a least fixpoint which we can take as the
{\em Kripke-Kleene fixpoint of} $T_{\mathsf{P}}$. Finally, as in the case of
first-order logic programs, the set of all fixpoints of ${\cal C}_{T_{\mathsf{P}}}$
is the set of {\em stable fixpoints of} $T_{\mathsf{P}}$, and can be taken as
the $3$-valued stable models of $\mathsf{P}$ (by Theorem~\ref{minimal-fixpoint}
in~\ref{appendix-of-section-6a}, every fixpoint of ${\cal C}_{T_{\mathsf{P}}}$
is also a fixpoint of $T_{\mathsf{P}}$ and therefore a model of $\mathsf{P}$).

In contrast to the above $3$-valued semantics, the definition of {\em $2$-valued stable models}
for higher-order logic programs seems less direct to obtain.
In the case of first-order logic programs, the $2$-valued stable models are those
fixpoints of ${\cal C}_{T_{\mathsf{P}}}$ that are {\em exact}~\cite{DMT00,DMT04},
i.e., that are of the form $(I,I)$. In the higher-order case however, things are not
that simple. Consider for example the positive higher-order logic program consisting
only of the rule {\tt p(R) $\leftarrow$ R}, where {\tt p} is of type $o\rightarrow o$.
Since this is a positive program, it is reasonable to assume that it has a unique
2-valued stable model which assigns to {\tt p} the identity relation over the
set of classical two truth values. The meaning of this program under the semantics
proposed in the present paper is captured by the pair of interpretations $(I,J)$ where:
$I({\tt p})(\mathit{false},\mathit{false})= \mathit{false}$,
$I({\tt p})(\mathit{true},\mathit{true})= \mathit{true}$,
$I({\tt p})(\mathit{false},\mathit{true})= \mathit{false}$, and
$J({\tt p})(\mathit{false},\mathit{false})= \mathit{false}$,
$J({\tt p})(\mathit{true},\mathit{true})= \mathit{true}$,
$J({\tt p})(\mathit{false},\mathit{true})= \mathit{true}$.
Notice that $I\neq J$ and this is due to the fact that $I$ and
$J$ are $3$-valued interpretations and not $2$-valued ones as
in the first-order case. In other words, under our semantics there
does not exist an exact pair of interpretations that is a fixpoint
of ${\cal C}_{T_{\mathsf{P}}}$ which we could take as the 2-valued
stable semantics of the program. What needs to be done here is to generalize
the notion of ``{\em exact pair of interpretations''}. Informally
speaking, a pair $(I,J)$ of Herbrand interpretations of $\mathsf{P}$
will be called exact if for every predicate constant $\mathsf{p}$
of the program, $I(\mathsf{p})$ coincides with $J(\mathsf{p})$
when they are applied to arguments that are {\em essentially $2$-valued}
(we need to define inductively for all types what it means for a relation
to be essentially $2$-valued). Notice that $I({\tt p})$ agrees with
$J({\tt p})$ when applied to $2$-valued arguments, i.e., when applied
to $(\mathit{true},\mathit{true})$ and $(\mathit{false},\mathit{false})$.
We believe that the approach sketched above leads to a characterization
of the $2$-valued stable models, but the details need to be carefully
examined and specified.

In this paper we have claimed that the proposed approach is an appealing
formulation for capturing the well-founded semantics for higher-order
logic programs with negation. We have substantiated our claim by demonstrating
that the proposed semantics generalizes the well-founded one for propositional
programs. As suggested by one of the reviewers, this claim would be stronger
if one could define alternative semantics that lead to the same model.
One such approach would be to extend the original definition of the
well-founded semantics~\cite{GelderRS91} which was based on the notion of
{\em unfounded sets}. Another promising direction would be to derive an
extension of Przymusinski's {\em iterated least fixpoint construction}~\cite{Prz89}
to the higher-order case. Both of these directions seem quite fruitful
and non-trivial, and certainly require further investigation.

\bibliographystyle{acmtrans}
\bibliography{iclp2018}

\label{lastpage}
\appendix

\iftplp
  \def\b@at{\begin{author@tabular}[t]{@{}c@{}}}

{\centering \sloppy
{\normalfont\LARGE\itshape {\large\textnormal{Online appendix for the paper}}   \\
Approximation Fixpoint Theory and the Well-Founded Semantics of Higher-Order Logic Programs
\\
{\large\textnormal{published in Theory and Practice of Logic Programming}}\par}%
\vspace{16pt}
{\normalfont\normalsize\rmfamily
    \b@at
    {ANGELOS CHARALAMBIDIS}
    \end{author@tabular}\par
    Institute of Informatics and Telecommunications, NCSR ``Demokritos'', Greece%
    \par
    \vspace{10pt}
    \b@at
      {PANOS RONDOGIANNIS, IOANNA SYMEONIDOU }
    \end{author@tabular}\par
    Department of Informatics and Telecommunications, University of Athens, Greece%
   \par}
  }%

\thispagestyle{myheadings}
\pagestyle{myheadings}
\fi

\section{$\!\!\!$: Mathematical Preliminaries and Proofs of Section~\ref{semantics_of_language}}\label{appendix-of-section-4}
A partially ordered set (or poset) $(L,\leq)$ is called a {\em lattice} if for all $x,y\in L$
there exists a least upper bound and a greatest lower bound. A lattice $(L,\leq)$ is called
{\em complete} if for all $S\subseteq L$, there exists a least upper bound and a
greatest lower bound, denoted by $\bigvee S$ and $\bigwedge S$ respectively.
Every complete lattice has a least element and a greatest element, denoted by
$\perp$ and $\top$ respectively. We will use the following two convenient equivalent
definitions of complete lattices~\cite[Theorem 2.31, page 47]{DP02}:
\begin{theorem}\label{alternative_lattice_definition}
A partially ordered set $(L,\leq)$ is a complete lattice if $L$ has a
least element and every non-empty subset $S\subseteq L$ has a least
upper bound in $L$. Alternatively, $(L,\leq)$ is a complete lattice if $L$ has a
greatest element and every non-empty subset $S\subseteq L$ has a greatest
lower bound in $L$.
\end{theorem}

Given a partially ordered set $(L,\leq)$, every linearly ordered subset $S$
of $L$ will be called a {\em chain}. A partially ordered set is {\em chain-complete}
if it has a least element $\perp$ and every chain $S\subseteq L$ has a least
upper bound.

\begingroup
\def\theproposition{\ref{semantics_of_types_lattice_cpo}}
\begin{proposition}
Let $D$ be a nonempty set. For every predicate type $\pi$, $(\lsem \pi \rsem_D, \leq_\pi)$
is a complete lattice and $(\lsem \pi \rsem_D, \preceq_\pi)$ is a chain complete poset.
\end{proposition}
\addtocounter{proposition}{-1}
\endgroup
\begin{proof}
Consider the first statement and let $\pi$ be an arbitrary predicate type.
Recall that $\perp_{\leq_{\pi}}$ exists; it suffices to show that for every
non-empty subset $S$ of $\lsem \pi \rsem_D$, the least upper bound of $S$
exists and belongs to $\lsem \pi \rsem_D$.

The least upper bound can be defined inductively on the structure of predicate types.
If $\pi =o$, then $\bigvee_{\leq_{o}}S$ is defined in the obvious way. For $\pi = \iota \rightarrow \pi_1$,
we define for all $d \in D$, $(\bigvee_{\leq_{\iota\rightarrow \pi_1}} S)(d) = \bigvee_{\leq_{\pi_1}}\{f(d) \mid f \in S\}$.
Finally, if $\pi = \pi_1 \rightarrow \pi_2$, we define for all $d \in \lsem \pi_1 \rsem_D$,
$(\bigvee_{\leq_{\pi_1\rightarrow \pi_2}} S)(d) = \bigvee_{\leq_{\pi_2}}\{f(d) \mid f \in S\}$.
We need to verify that for type $\pi_1\rightarrow \pi_2$ the least upper bound is a Fitting-monotonic function.
This is a consequence of the following auxiliary statement, which we need to establish for every predicate type $\pi$:

\vspace{0.2cm}
\noindent
{\em Auxiliary statement:} Let $I$ be a non-empty index-set and let $d_i,d'_i\in \lsem \pi \rsem_D$, $i \in I$.
If for all $i \in I$, $d_i \preceq_{\pi} d'_i$, then $\bigvee_{\leq_{\pi}}\{ d_i \mid i \in I \} \preceq_{\pi} \bigvee_{\leq_{\pi}}\{ d'_i \mid i \in I \}$.

\vspace{0.2cm}
The proof of the auxiliary statement is by a simple induction on the structure of $\pi$. For type $\pi=o$ the
statement follows by a case analysis on the value of $\bigvee_{\leq_{\pi}}\{ d_i \mid i \in I \}$.
For types $\iota \rightarrow \pi_1$ and $\pi_1 \rightarrow \pi_2$,
the statement follows directly by the induction hypothesis. The auxiliary statement implies that
$(\bigvee_{\leq_{\pi_1\rightarrow \pi_2}} S)$ is a Fitting-monotonic function. More specifically,
for all $d,d'\in \lsem \pi_1 \rsem_D$ with $d\preceq_{\pi_1} d'$, it holds $f(d) \preceq_{\pi_2} f(d')$ for
every $f\in S$ (because the members of $S$ are Fitting-monotonic functions). Then, the auxiliary statement implies that
$\bigvee_{\leq_{\pi_2}}\{ f(d) \mid f \in S \} \preceq_{\pi_2} \bigvee_{\leq_{\pi_2}}\{ f(d) \mid f \in S\}$
which is equivalent to $(\bigvee_{\leq_{\pi_1\rightarrow \pi_2}} S)(d) \preceq_{\pi_2} (\bigvee_{\leq_{\pi_1\rightarrow \pi_2}} S)(d')$,
which means that $(\bigvee_{\leq_{\pi_1\rightarrow \pi_2}} S)$ is Fitting-monotonic.

Consider now the second statement. Notice that $(\lsem \pi \rsem_D, \preceq_\pi)$ is not
a complete lattice (for example, the set $\{\mfalse,\mtrue\}$ does not have a least upper
bound with respect to $\preceq_o$). However, it is a chain complete poset. For every type $\pi$, $\perp_{\preceq_{\pi}}$ exists.
Moreover, given a chain $S$ of elements of $\lsem \pi \rsem_D$, it suffices to verify that
$\bigvee_{\preceq_{\pi}}S$ exists and belongs to $\lsem \pi \rsem_D$. The proof is by induction
on the structure of $\pi$. For type $\pi=o$ it is obvious. For $\pi = \iota \rightarrow \pi_1$,
define $(\bigvee_{\preceq_{\iota\rightarrow \pi_1}} S)(d) = \bigvee_{\preceq_{\pi_1}}\{f(d) \mid f \in S\}$.
For $\pi = \pi_1 \rightarrow \pi_2$
define $(\bigvee_{\preceq_{\pi_1\rightarrow \pi_2}} S)(d) = \bigvee_{\preceq_{\pi_2}}\{f(d) \mid f \in S\}$.
We need to verify that $(\bigvee_{\preceq_{\pi_1\rightarrow \pi_2}} S)$ is a Fitting-monotonic
function, i.e., that for all $d,d' \in \lsem \pi_1 \rsem_D$ with $d\preceq_{\pi_1} d'$, it holds
$(\bigvee_{\preceq_{\pi_1\rightarrow \pi_2}} S)(d) \preceq_{\pi_2} (\bigvee_{\preceq_{\pi_1\rightarrow \pi_2}} S)(d')$,
or equivalently that
$\bigvee_{\preceq_{\pi_2}}\{f(d)\mid f \in S\} \preceq_{\pi_2} \bigvee_{\preceq_{\pi_2}}\{f(d')\mid f \in S\}$,
which holds because for every $f \in S$, $f(d) \preceq_{\pi_2} f(d')$.
\end{proof}

The proof of the above lemma has as a direct consequence the following corollary:
\begin{corollary}\label{prop-lub-preceq}
Let $D$ be a nonempty set and $\pi$ a predicate type. Let $I$ be a non-empty index-set and let $d_i,d'_i\in \lsem \pi \rsem_D$, $i \in I$.
If for all $i \in I$, $d_i \preceq_{\pi} d'_i$, then $\bigvee_{\leq_{\pi}}\{ d_i \mid i \in I \} \preceq_{\pi} \bigvee_{\leq_{\pi}}\{ d'_i \mid i \in I \}$.
\end{corollary}

\section{$\!\!\!$: Proofs of Section~\ref{bijection}}\label{appendix-of-section-5}
\begingroup
\def\theproposition{\ref{ma-am-lattices}}
\begin{proposition}
Let $D$ be a nonempty set. For every predicate type $\pi$, $(\lsem \pi \rsem_D^{\mathsf{ma}}, \leq_\pi)$ and
$(\lsem \pi \rsem_D^{\mathsf{am}}, \leq_\pi)$ are complete lattices.
\end{proposition}
\addtocounter{proposition}{-1}
\endgroup
\begin{proof}
We give the proof for the case $(\lsem \pi \rsem_D^{\mathsf{ma}}, \leq_\pi)$; the
case $(\lsem \pi \rsem_D^{\mathsf{am}}, \leq_\pi)$ is symmetrical and omitted.
The proof is by induction on the structure of $\pi$. For $\pi=o$ the result
is immediate. We show the result for types $\iota \rightarrow \pi$ and
$\pi_1 \rightarrow \pi_2$, assuming it holds for $\pi$, $\pi_1$ and $\pi_2$.

Consider first the set  $\lsem \iota \rightarrow \pi \rsem_D^\mathsf{ma} = D\rightarrow \lsem \pi \rsem^\mathsf{ma}$.
This set has a least element, namely the function that assigns to each $d \in D$
the bottom element of type $\pi$. Let $S \subseteq  D\rightarrow \lsem \pi \rsem^\mathsf{ma}$ be
a nonempty set. For every $d\in D$ we define
$(\bigvee_{\leq_{\iota \rightarrow \pi}} S)(d)= \bigvee_{\leq_{\pi}}\{f(d) \mid f \in S\}$, which by the induction hypothesis
exists and belongs to $\lsem \pi \rsem_D^{\mathsf{ma}}$.

Consider now the set $\lsem \pi_1 \rightarrow \pi_2 \rsem_D^\mathsf{ma} =
[ (\lsem \pi_1 \rsem_D^\mathsf{ma} \otimes \lsem \pi_1 \rsem_D^\mathsf{am}) \stackrel{\mathsf{ma}}{\rightarrow} \lsem \pi_2 \rsem_D^\mathsf{ma}]$.
This set has a least element, namely the function that assigns to each pair
$(x,y) \in (\lsem \pi_1 \rsem_D^\mathsf{ma} \otimes \lsem \pi_1 \rsem_D^\mathsf{am})$
the bottom element of type $\perp_{\pi_{2}}$; this function is constant and therefore obviously
monotone-antimonotone.  Let $S \subseteq  [ (\lsem \pi_1 \rsem_D^\mathsf{ma} \otimes \lsem \pi_1 \rsem_D^\mathsf{am})
\stackrel{\mathsf{ma}}{\rightarrow} \lsem \pi_2 \rsem_D^\mathsf{ma}]$ be a nonempty set.
For every $(x,y) \in (\lsem \pi_1 \rsem_D^\mathsf{ma} \otimes \lsem \pi_1 \rsem_D^\mathsf{am})$
we define $(\bigvee_{\leq_{\pi_1 \rightarrow \pi_2}} S)(x,y)= \bigvee_{\leq_{\pi_2}}\{f(x,y) \mid f \in S\}$, which by the induction hypothesis
exists and belongs to $\lsem \pi_2 \rsem_D^{\mathsf{ma}}$. It remains to show that
$\bigvee S$ is monotone-antimonotone. Consider $(x,y),(x',y') \in (\lsem \pi_1 \rsem_D^\mathsf{ma} \otimes \lsem \pi_1 \rsem_D^\mathsf{am})$
and assume that $x \leq x'$ and $y \geq y'$. It suffices to show that
$(\bigvee_{\leq_{\pi_1\rightarrow \pi_2}} S)(x,y) \leq_{\pi_2} (\bigvee_{\leq_{\pi_1\rightarrow \pi_2}} S)(x',y')$. Since every element
of $S$ is monotone-antimonotone, for every $f \in S$ it holds
$f(x,y) \leq_{\pi_2} f(x',y')$. Therefore, $\bigvee_{\leq_{\pi_2}}\{f(x,y) \mid f \in S\} \leq_{\pi_2} \bigvee_{\leq_{\pi_2}}\{f(x',y') \mid f \in S\}$,
and thus $(\bigvee S_{\leq_{\pi_1\rightarrow \pi_2}})(x,y) \leq_{\pi_2} (\bigvee S_{\leq_{\pi_1\rightarrow \pi_2}})(x',y')$.
\end{proof}

The proof of Proposition~\ref{pairs_complete_lattice_cpo} requires the following lemma
which can be established by induction on the structure of $\pi$:
\begin{lemma}\label{interlattice_lub_lemma}
Let $D$ be a nonempty set and let $\pi$ be a predicate type. Let $S \subseteq \lsem \pi \rsem^{\mathsf{ma}}_D$
and $g \in \lsem \pi \rsem^{\mathsf{am}}_D$.
\begin{itemize}
\item If for all $f \in S$, $f \leq g$,  then $\bigvee S \leq g$.

\item If for all $f \in S$, $f\geq g$, then $\bigwedge S \geq g$.
\end{itemize}
\end{lemma}
\begingroup
\def\theproposition{\ref{pairs_complete_lattice_cpo}}
\begin{proposition}
Let $D$ be a nonempty set. For each predicate type $\pi$, $\lsem \pi \rsem_D^{\mathsf{ma}} \otimes \lsem \pi \rsem_D^{\mathsf{am}}$
is a complete lattice with respect to $\leq_{\pi}$ and a chain-complete poset with respect to $\preceq_{\pi}$.
\end{proposition}
\addtocounter{proposition}{-1}
\endgroup
\begin{proof}
For every $\pi$ it is straightforward to define the bottom elements of the partially
ordered sets $(\lsem \pi \rsem_D^{\mathsf{ma}} \otimes \lsem \pi \rsem_D^{\mathsf{am}}, \leq_{\pi})$
and $(\lsem \pi \rsem_D^{\mathsf{ma}} \otimes \lsem \pi \rsem_D^{\mathsf{am}}, \preceq_{\pi})$.

Given $S \subseteq \lsem \pi \rsem_D^{\mathsf{ma}} \otimes \lsem \pi \rsem_D^{\mathsf{am}}$,
we define $\bigvee_{\leq_{\pi}} S = (\bigvee_{\leq_{\pi}} \{f \mid (f,g) \in S\}, \bigvee_{\leq_{\pi}} \{g \mid (f,g) \in S\})$.
It can be easily seen that $\bigvee_{\leq_{\pi}} S \in \lsem \pi \rsem_D^{\mathsf{ma}} \otimes \lsem \pi \rsem_D^{\mathsf{am}}$
due to Proposition~\ref{ma-am-lattices}, Lemma~\ref{interlattice_lub_lemma} and the fact that for every pair $(f,g) \in S$, $f \leq_{\pi} g$.

On the other hand, let $S \subseteq \lsem \pi \rsem_D^{\mathsf{ma}} \otimes \lsem \pi \rsem_D^{\mathsf{am}}$
be a chain. We define
$\bigvee_{\preceq_{\pi}} S = (\bigvee_{\leq_{\pi}} \{f \mid (f,g) \in S\}, \bigwedge_{\leq_{\pi}} \{g \mid (f,g) \in S\})$.
It is straightforward to show that $\bigvee_{\preceq_{\pi}} S$ is the $\preceq_{\pi}$-least upper bound of the chain.
Moreover, $(\bigvee_{\preceq_{\pi}} S) \in \lsem \pi \rsem_D^{\mathsf{ma}} \otimes \lsem \pi \rsem_D^{\mathsf{am}}$
because $\bigvee_{\leq_{\pi}} \{f \mid (f,g) \in S\} \leq_{\pi} \bigwedge_{\leq_{\pi}} \{g \mid (f,g) \in S\}$
(this can easily be shown using basic properties of lubs and glbs, Lemma~\ref{interlattice_lub_lemma},
and the fact that $S$ is a chain; see also Proposition 2.3 in~\cite{DMT04}).
\end{proof}
\begingroup
\def\theproposition{\ref{tau-monotonic}}
\begin{proposition}
Let $D$ be a nonempty set and let $\pi$ be a predicate type. Then, for every $f,g \in \lsem \pi \rsem_D$
and for every $(f_1,f_2),(g_1,g_2) \in \lsem \pi \rsem_D^{\mathsf{ma}}\otimes \lsem \pi \rsem_D^{\mathsf{am}}$,
the following statements hold:
  \begin{enumerate}
  \item $\tau_{\pi}(f) \in  (\lsem \pi \rsem_D^{\mathsf{ma}} \otimes \lsem \pi \rsem_D^{\mathsf{am}})$
        and $\tau^{-1}_{\pi}(f_1,f_2) \in \lsem \pi \rsem_D$.
  \item If $f \preceq_{\pi} g$ then $\tau_\pi(f) \preceq_{\pi} \tau_\pi(g)$.
  \item If $f \leq_{\pi} g$ then $\tau_\pi(f) \leq_{\pi} \tau_\pi(g)$.
  \item If $(f_1,f_2) \preceq_{\pi} (g_1,g_2)$ then $\tau_\pi^{-1}(f_1,f_2) \preceq_{\pi} \tau_\pi^{-1}(g_1,g_2)$.
  \item If $(f_1,f_2) \leq_{\pi} (g_1,g_2)$ then $\tau^{-1}_\pi(f_1,f_2) \leq_{\pi} \tau^{-1}_\pi(g_1,g_2)$.
\end{enumerate}
\end{proposition}
\addtocounter{proposition}{-1}
\endgroup
\begin{proof}
The five statements are shown by a simultaneous induction on the structure of $\pi$.
We give the proofs for Statement 1, Statement 2 (the proof of Statement 3 is analogous
and omitted) and Statement 4 (the proof of Statement 5 is similar and omitted).

The basis case is for $\pi=o$ and is straightforward for all statements. We assume
the statements hold for $\pi$, $\pi_1$ and $\pi_2$. We demonstrate that they hold for
$\iota \rightarrow \pi$ and for $\pi_1 \rightarrow \pi_2$.

\vspace{0.2cm}
\noindent
{\em Statement 1:}  Consider first the case of $\iota\rightarrow \pi$. It suffices to show that
$\tau_{\iota \rightarrow \pi}(f) \in (\lsem \iota \rightarrow \pi \rsem_D^{\mathsf{ma}} \otimes \lsem \iota \rightarrow \pi \rsem_D^{\mathsf{am}})$.
By the induction hypothesis, $\tau_{\pi}(f(d)) \in (\lsem \pi \rsem_D^{\mathsf{ma}} \otimes \lsem \pi \rsem_D^{\mathsf{am}})$.
Therefore, $[\tau_\pi(f(d))]_1 \leq [\tau_\pi(f(d))]_2$, and consequently
$(\lambda d. [\tau_\pi(f(d))]_1, \lambda d. [\tau_\pi(f(d))]_2) \in
(\lsem \iota \rightarrow \pi \rsem_D^{\mathsf{ma}} \otimes \lsem \iota \rightarrow \pi \rsem_D^{\mathsf{am}})$.
We next show that $\tau^{-1}_{\pi}(f_1,f_2) \in \lsem \iota \rightarrow \pi \rsem_D$.
Since $(f_1,f_2)\in (\lsem \iota \rightarrow \pi \rsem_D^{\mathsf{ma}} \otimes \lsem \iota \rightarrow \pi \rsem_D^{\mathsf{am}})$,
$f_1 \leq f_2$ and $(f_1(d),f_2(d))\in (\lsem \pi \rsem_D^{\mathsf{ma}} \otimes \lsem \pi \rsem_D^{\mathsf{am}})$.
By the induction hypothesis, $\tau_{\pi}^{-1}(f_1(d),f_2(d)) \in \lsem \pi \rsem_D$
and $\lambda d.\tau_{\pi}^{-1}(f_1(d),f_2(d)) \in \lsem \iota \rightarrow \pi \rsem_D$.

Consider the case $\pi_1\rightarrow \pi_2$. We show that $\tau_{\pi_1\rightarrow \pi_2}(f) \in(\lsem
\pi_1\rightarrow \pi_2 \rsem^{\mathsf{ma}}_D \otimes \lsem \pi_1\rightarrow \pi_2 \rsem^{\mathsf{am}}_D)$.
Let $(d_1,d_2) \in (\lsem\pi_1\rsem^{\mathsf{ma}}_D \otimes \lsem \pi_1\rsem^{\mathsf{am}}_D)$.
By the induction hypothesis $\tau_{\pi_1}^{-1}(d_1,d_2) \in \lsem \pi_1 \rsem_D$,
$f(\tau_{\pi_1}^{-1}(d_1,d_2)) \in \lsem \pi_2 \rsem_D$,
and $\tau_{\pi_2}(f(\tau_{\pi_1}^{-1}(d_1,d_2))) \in (\lsem\pi_1\rsem^{\mathsf{ma}}_D \otimes \lsem \pi_1\rsem^{\mathsf{am}}_D)$,
which has as a direct consequence that $[\tau_{\pi_2}(f(\tau_{\pi_1}^{-1}(d_1,d_2)))]_1 \leq [\tau_{\pi_2}(f(\tau_{\pi_1}^{-1}(d_1,d_2)))]_2$.
Therefore, $\lambda(d_1,d_2).[\tau_{\pi_2}(f(\tau_{\pi_1}^{-1}(d_1,d_2)))]_1 \leq
\lambda(d_1,d_2).[\tau_{\pi_2}(f(\tau_{\pi_1}^{-1}(d_1,d_2)))]_2$. It remains to
show that the function $\lambda(d_1,d_2).[\tau_{\pi_2}(f(\tau_{\pi_1}^{-1}(d_1,d_2)))]_1$
is monotone-antimonotone and the function $\lambda(d_1,d_2).[\tau_{\pi_2}(f(\tau_{\pi_1}^{-1}(d_1,d_2)))]_2$
is antimonotone-monotone. This follows by using the induction hypothesis
for Statement 4, the Fitting-monotonicity of $f$, and the induction hypothesis of Statement 2.
The fact that $\tau^{-1}_{\pi_1\rightarrow\pi_2}(f_1,f_2) \in \lsem \pi_1 \rightarrow \pi_2\rsem_D$
follows using similar arguments as above.

\vspace{0.2cm}
\noindent
{\em Statement 2:}  Consider first the case of $\iota\rightarrow \pi$.
It suffices to show that:
$$(\lambda d. [\tau_\pi(f(d))]_1, \lambda d. [\tau_\pi(f(d))]_2) \preceq (\lambda d. [\tau_\pi(g(d))]_1, \lambda d. [\tau_\pi(g(d))]_2)$$
or equivalently that $\lambda d. [\tau_\pi(f(d))]_1 \leq \lambda d. [\tau_\pi(g(d))]_1$
and $\lambda d. [\tau_\pi(f(d))]_2 \geq \lambda d. [\tau_\pi(g(d))]_2$, or equivalently that
for every $d$, $[\tau_\pi(f(d))]_1 \leq [\tau_\pi(g(d))]_1$ and $[\tau_\pi(f(d))]_2 \geq [\tau_\pi(g(d))]_2$.
This holds because, since $f\preceq g$, it holds $f(d) \preceq g(d)$ and by the induction hypothesis,
$\tau_\pi(f(d)) \preceq \tau_\pi(g(d))$. Consider now the case of $\pi_1 \rightarrow \pi_2$.
It suffices to show that:
\[
\begin{array}{l}
(\lambda(d_1, d_2). [\tau_{\pi_2}(f(\tau^{-1}_{\pi_1}(d_1, d_2)))]_1, \lambda(d_1, d_2). [\tau_{\pi_2}(f(\tau^{-1}_{\pi_1}(d_1, d_2)))]_2)\preceq \\
(\lambda(d_1, d_2). [\tau_{\pi_2}(g(\tau^{-1}_{\pi_1}(d_1, d_2)))]_1, \lambda(d_1, d_2). [\tau_{\pi_2}(g(\tau^{-1}_{\pi_1}(d_1, d_2)))]_2)
\end{array}
\]
or equivalently that $\lambda(d_1, d_2). [\tau_{\pi_2}(f(\tau^{-1}_{\pi_1}(d_1, d_2)))]_1 \leq \lambda(d_1, d_2). [\tau_{\pi_2}(g(\tau^{-1}_{\pi_1}(d_1, d_2)))]_1$
and $\lambda(d_1, d_2). [\tau_{\pi_2}(f(\tau^{-1}_{\pi_1}(d_1, d_2)))]_2 \geq \lambda(d_1, d_2). [\tau_{\pi_2}(g(\tau^{-1}_{\pi_1}(d_1, d_2)))]_2$,
or equivalently that for all $d_1,d_2$, $[\tau_{\pi_2}(f(\tau^{-1}_{\pi_1}(d_1, d_2)))]_1 \leq [\tau_{\pi_2}(g(\tau^{-1}_{\pi_1}(d_1, d_2)))]_1$
and $[\tau_{\pi_2}(f(\tau^{-1}_{\pi_1}(d_1, d_2)))]_2 \geq [\tau_{\pi_2}(g(\tau^{-1}_{\pi_1}(d_1, d_2)))]_2$.
Since $f\preceq g$, it holds that $f(\tau^{-1}_{\pi_1}(d_1, d_2))) \preceq g(\tau^{-1}_{\pi_1}(d_1, d_2)))$,
and by the induction hypothesis $\tau_{\pi_2}(f(\tau^{-1}_{\pi_1}(d_1, d_2))) \preceq \tau_{\pi_2}(g(\tau^{-1}_{\pi_1}(d_1, d_2)))$,
which is the desired result.

\vspace{0.2cm}
\noindent
{\em Statement 4:} Consider first the case of $\iota\rightarrow \pi$.
It suffices to show that:
$$\lambda d. \tau_\pi^{-1}(f_1(d), f_2(d))  \preceq \lambda d. \tau_\pi^{-1}(g_1(d), g_2(d))$$
or equivalently that for every $d$, $\tau_\pi^{-1}(f_1(d), f_2(d))  \preceq  \tau_\pi^{-1}(g_1(d), g_2(d))$.
Since $(f_1,f_2) \preceq (g_1,g_2)$, it holds $(f_1(d),f_2(d)) \preceq (g_1(d),g_2(d))$, and the
result follows from the induction hypothesis. Consider now the case of $\pi_1 \rightarrow \pi_2$.
It suffices to show that:
$$\lambda d. \tau^{-1}_{\pi_2}(f_1(\tau_{\pi_1}(d)), f_2(\tau_{\pi_1}(d)))\preceq \lambda d. \tau^{-1}_{\pi_2}(g_1(\tau_{\pi_1}(d)), g_2(\tau_{\pi_1}(d)))$$
or equivalently that for every $d$, $\tau^{-1}_{\pi_2}(f_1(\tau_{\pi_1}(d)), f_2(\tau_{\pi_1}(d)))\preceq \tau^{-1}_{\pi_2}(g_1(\tau_{\pi_1}(d)), g_2(\tau_{\pi_1}(d)))$. Since $(f_1,f_2) \preceq (g_1,g_2)$, it holds $(f_1(\tau_{\pi_1}(d)),f_2(\tau_{\pi_1}(d))) \preceq (g_1(\tau_{\pi_1}(d)),g_2(\tau_{\pi_1}(d)))$, and the result follows from the induction hypothesis.
\end{proof}
\begingroup
\def\theproposition{\ref{tau-elim}}
\begin{proposition}
Let $D$ be a nonempty set and let $\pi$ be a predicate type. Then, for every $f \in \lsem \pi \rsem_D$, $\tau^{-1}_{\pi}(\tau_{\pi}(f))=f$,
and for every $(f_1,f_2) \in \lsem \pi \rsem_D^{\mathsf{ma}}\otimes \lsem \pi \rsem_D^{\mathsf{am}}$,
$\tau_{\pi}(\tau^{-1}_{\pi}(f_1,f_2))=(f_1,f_2)$.
\end{proposition}
\addtocounter{proposition}{-1}
\endgroup
\begin{proof}
The proof of the two statements is by a simultaneous induction on the structure of $\pi$.
The case $\pi = o$ is immediate. Assume the two statements hold for $\pi$, $\pi_1$ and $\pi_2$.
We demonstrate that they hold for $\iota \rightarrow \pi$ and for $\pi_1 \rightarrow \pi_2$.

We have:
\[
\begin{array}{lll}
                                                                  &   & \tau^{-1}_{\iota \rightarrow \pi}(\tau_{\iota\rightarrow \pi}(f)) \,\,=\vspace{0.1cm}\\
                                                                  & = & \tau^{-1}_{\iota\rightarrow \pi}(\lambda d. [\tau_\pi(f(d))]_1, \lambda d. [\tau_\pi(f(d))]_2)\\
                                                                  &   &  \mbox{(Definition of $\tau_{\iota \rightarrow \pi}$)}\\
                                                                  & = & \lambda d. \tau^{-1}_{\pi}([\tau_\pi(f(d))]_1,  [\tau_\pi(f(d))]_2)\\
                                                                  &   & \mbox{(Definition of $\tau^{-1}_{\iota \rightarrow \pi}$)}\\
                                                                  & = & \lambda d. \tau^{-1}_{\pi}(\tau_{\pi}(f(d)))\\
                                                                  &   & \mbox{(Definition of $[\cdot]_1$ and $[\cdot]_2$)}\\
                                                                  & = & \lambda d. f(d) \\
                                                                  &   & \mbox{(Induction Hypothesis)}\\
                                                                  & = & f
\end{array}
\]
Also:
\[
\begin{array}{lll}
                                                                  &   & \tau_{\iota \rightarrow \pi}(\tau^{-1}_{\iota\rightarrow \pi}(f_1,f_2))\,\,=\vspace{0.1cm}\\
                                                                  & = & \tau_{\iota\rightarrow \pi}(\lambda d. \tau_{\pi}^{-1}(f_1(d),f_2(d))) \\
                                                                        &   & \mbox{(Definition of $\tau^{-1}_{\iota \rightarrow \pi}$)}\\
                                                                  & = & (\lambda d. [\tau_{\pi}(\tau^{-1}_\pi(f_1(d),f_2(d)))]_1, \lambda d. [\tau_{\pi}(\tau^{-1}_\pi(f_1(d),f_2(d)))]_2) \\
                                                                  &   & \mbox{(Definition of $\tau_{\iota \rightarrow \pi}$)}\\
                                                                  & = & (\lambda d. [(f_1(d),f_2(d))]_1, \lambda d. [(f_1(d),f_2(d))]_2) \\
                                                                  &   & \mbox{(Induction Hypothesis)}\\
                                                                  & = & (\lambda d. f_1(d),\lambda d. f_2(d)) \\
                                                                  &   & \mbox{(Definition of $[\cdot]_1$ and $[\cdot]_2$)}\\
                                                                  & = & (f_1,f_2)
\end{array}
\]
Consider now the case of $\pi_1\rightarrow \pi_2$. We have:
\[
\begin{array}{lll}
                                                                  &   & \tau^{-1}_{\pi_1 \rightarrow \pi_2}(\tau_{\pi_1\rightarrow \pi_2}(f))\,\,=\vspace{0.1cm}\\
                                                                  & = & \tau^{-1}_{\pi_1\rightarrow \pi_2}(\lambda(d_1, d_2). [\tau_{\pi_2}(f(\tau^{-1}_{\pi_1}(d_1, d_2)))]_1,
                                                                                   \lambda(d_1, d_2). [\tau_{\pi_2}(f(\tau^{-1}_{\pi_1}(d_1, d_2)))]_2)\\
                                                                  &   &  \mbox{(Definition of $\tau_{\pi_1 \rightarrow \pi_2}$)}\\
                                                                  & = & \lambda d. \tau^{-1}_{\pi_2}([\tau_{\pi_2}(f(\tau^{-1}_{\pi_1}(\tau_{\pi_1}(d))))]_1,
                                                                                                     [\tau_{\pi_2}(f(\tau^{-1}_{\pi_1}(\tau_{\pi_1}(d))))]_2)\\
                                                                  &   & \mbox{(Definition of $\tau^{-1}_{\pi_1 \rightarrow \pi_2}$)}\\
                                                                  & = & \lambda d. \tau^{-1}_{\pi_2}([\tau_{\pi_2}(f(d))]_1, [\tau_{\pi_2}(f(d))]_2)\\
                                                                  &   & \mbox{(Induction Hypothesis)}\\
                                                                  & = & \lambda d. \tau^{-1}_{\pi_2}(\tau_{\pi_2}(f(d)))\\
                                                                  &   & \mbox{(Definition of $[\cdot]_1$ and $[\cdot]_2$)}\\
                                                                  & = & \lambda d. f(d)\\
                                                                  &   & \mbox{(Induction Hypothesis)}\\
                                                                  & = & f
\end{array}
\]
Also:
\[
\begin{array}{lll}
                                                                  &   & \tau_{\pi_1 \rightarrow \pi_2}(\tau^{-1}_{\pi_1\rightarrow \pi_2}(f_1,f_2))\,\,=\vspace{0.1cm}\\
                                                                  & = & \tau_{\pi_1\rightarrow \pi_2}(\lambda d. \tau_{\pi_2}^{-1}(f_1(\tau_{\pi_1}(d)),f_2(\tau_{\pi_1}(d))) \\
                                                                        &   & \mbox{(Definition of $\tau^{-1}_{\pi_1 \rightarrow \pi_2}$)}\\
                                                                  & = & (\lambda(d_1, d_2). [\tau_{\pi_2}(\tau_{\pi_2}^{-1}(f_1(\tau_{\pi_1}(\tau^{-1}_{\pi_1}(d_1, d_2))),f_2(\tau_{\pi_1}(\tau^{-1}_{\pi_1}(d_1, d_2)))))]_1, \\
                                                                  &   & \,\,\lambda(d_1, d_2). [\tau_{\pi_2}(\tau_{\pi_2}^{-1}(f_1(\tau_{\pi_1}(\tau^{-1}_{\pi_1}(d_1, d_2))),f_2(\tau_{\pi_1}(\tau^{-1}_{\pi_1}(d_1, d_2)))))]_2) \\
                                                                  &   & \mbox{(Definition of $\tau_{\pi_1 \rightarrow \pi_2}$)}\\
                                                                  & = & (\lambda(d_1, d_2). [f_1(d_1,d_2),f_2(d_1,d_2)]_1,\lambda(d_1, d_2). [f_1(d_1,d_2),f_2(d_1,d_2)]_2)\\
                                                                  &   & \mbox{(Induction Hypothesis)}\\
                                                                  & = & (\lambda(d_1, d_2). f_1(d_1,d_2),\lambda(d_1, d_2).f_2(d_1,d_2))\\
                                                                  &   & \mbox{(Definition of $[\cdot]_1$ and $[\cdot]_2$)}\\
                                                                  & = & (f_1,f_2)
\end{array}
\]
The above completes the proof of the proposition.
\end{proof}

\section{$\!\!\!$: An Extension of Consistent Approximation Fixpoint Theory}\label{appendix-of-section-6a}
In this appendix we propose a mild extension of the theory of consistent
approximating operators developed in~\cite{DMT04}. We briefly highlight the main
idea behind the work in~\cite{DMT04} and then justify the necessity for our extension.

Let $(L,\leq)$ be a complete lattice. The authors in~\cite{DMT04} consider the set
$L^c=\{(x,y) \in L \times L \mid x\leq y\}$. Intuitively speaking, a pair $(x,y) \in L^c$
can be viewed as an approximation to all elements $z \in L$ such that $x\leq z \leq y$.
An operator $A:L^c \rightarrow L^c$ is called in~\cite{DMT04} a {\em consistent approximating
operator} if it is $\preceq$-monotone (see below) and for every $x \in L$, $A(x,x)_1=A(x,x)_2$
(the subscripts 1 and 2 denote projection to the first and second elements respectively of the
pair returned by $A$). In Section 3 of~\cite{DMT04}, an elegant theory is developed
whose purpose is to demonstrate how, under specific conditions, one can characterize the
{\em well-founded fixpoint} of a given consistent approximating operator $A$. Since approximating
operators emerge in many non-monotonic formalisms, the theory developed in~\cite{DMT04} provides a
useful tool for the study of the semantics of such formalisms.

In our work, the immediate consequence operator $T_{\mathsf{P}}$ is not an approximating
operator in the sense of~\cite{DMT04}. More specifically, $T_{\mathsf{P}}$ is a function in
$({\cal H}^{\mathsf{ma}}_\mathsf{P} \otimes  {\cal H}^{\mathsf{am}}_\mathsf{P})
\rightarrow ({\cal H}^{\mathsf{ma}}_\mathsf{P} \otimes  {\cal H}^{\mathsf{am}}_\mathsf{P})$.
In other words, there is not just a single lattice $L$ involved in the definition of
$T_{\mathsf{P}}$, but instead two lattices, namely ${\cal H}^{\mathsf{ma}}_{\mathsf{P}}$
and ${\cal H}^{\mathsf{am}}_\mathsf{P}$. Moreover, the condition
``for every $x \in L$, $A(x,x)_1=A(x,x)_2$'' required in~\cite{DMT04}, does not hold
in our case, because the two arguments of $T_{\mathsf{P}}$ range over two different sets
(namely ${\cal H}^{\mathsf{ma}}_{\mathsf{P}}$ and ${\cal H}^{\mathsf{am}}_\mathsf{P}$).
We therefore need to define an extension of the material in Section 3 of~\cite{DMT04},
that suits our purposes.

In the following, we develop the above mentioned extension following closely the
statements and proofs of~\cite{DMT04}. The material is presented in an abstract form
(as in~\cite{DMT04}), with the purpose of having a wider applicability than the
present paper. In order to retrieve the connections with the present paper,
one can take $A=T_{\mathsf{P}}$, $L_1 ={\cal H}^{\mathsf{ma}}_\mathsf{P}$
and $L_2 = {\cal H}^{\mathsf{am}}_\mathsf{P}$.

Let $(L,\leq)$ be a partially ordered set and assume that $L$ contains a least element
$\perp$ and a greatest element $\top$ with respect to $\leq$. Let $L_1,L_2 \subseteq L$
be non-empty sets such that $L_1 \cup L_2 = L$ and $(L_1,\leq)$ and $(L_2,\leq)$ are
complete lattices that both contain the elements $\perp$ and $\top$. We will denote
the least upper bound operations in the two lattices by $\textit{lub}_{L_1}$ and
$\textit{lub}_{L_2}$ respectively (we will also use $\bigvee_{L_1}$ and $\bigvee_{L_2}$).
We denote the greatest lower bound operations by $\textit{glb}_{L_1}$ and $\textit{glb}_{L_2}$
(also denoted by $\bigwedge_{L_1}$ and $\bigwedge_{L_2}$).
We assume that our lattices satisfy the following two properties:

\begin{enumerate}
\item {\bf Interlattice Lub Property:} Let $b \in L_2$ and $S \subseteq L_1$ such that for every $x\in S$,
$x\leq b$. Then, $\bigvee_{L_1} S \leq b$.

\item {\bf Interlattice Glb Property:} Let $a \in L_1$ and $S \subseteq L_2$ such that for every $x\in S$,
$x\geq a$. Then, $\bigwedge_{L_2} S \geq a$.
\end{enumerate}

\noindent
{\bf Remark:} It can be easily verified (see Lemma~\ref{interlattice_lub_lemma} in~\ref{appendix-of-section-5})
that both the Interlattice Lub Property and the Interlattice Glb Property hold when we take
$L_1={\cal H}^{\mathsf{ma}}_{\mathsf{P}}$ and $L_2 = {\cal H}^{\mathsf{am}}_{\mathsf{P}}$.

\vspace{0.2cm}

Given $(x,y),(x',y')\in L_1 \times L_2$, we will write
$(x,y) \preceq (x',y')$ if $x\leq x'$ and $y' \leq y$. We will write:
$$L_1 \otimes L_2 = \{(x,y) \mid x \in L_1, y\in L_2, x\leq y\}$$
The above set is non-empty since $(\perp,\top)\in L_1 \otimes L_2$.
\begin{definition}
A function $A: L_1 \otimes L_2 \rightarrow L_1 \otimes L_2$ is called a {\em consistent approximating operator}
if it is $\preceq$-monotonic.
\end{definition}

We will write $\textit{Appx}(L_1\otimes L_2)$ for the set of all consistent approximating
operators over $L_1\otimes L_2$. In the following results we assume we work with a given
consistent approximating operator $A$ (and therefore the symbol $A$ will appear free in
most definitions and results).
\begin{definition}
The pair $(a,b) \in L_1 \otimes L_2$ will be called $A$-reliable if $(a,b)\preceq A(a,b)$.
\end{definition}
Given $a \in L_1$ and $b\in L_2$, we write $[a,b]_{L_1}= \{x \in L_1\mid a \leq x \leq b\}$.
Symmetrically, $[a,b]_{L_2}= \{x \in L_2\mid a \leq x \leq b\}$.

\begin{proposition}\label{interals-complete-lattices}
For all $a \in L_1$ and $b\in L_2$, the sets $[\perp,b]_{L_1}$ and $[a,\top]_{L_2}$ are complete lattices.
\end{proposition}
\begin{proof}
We use Theorem~\ref{alternative_lattice_definition} of~\ref{appendix-of-section-4}.
Consider first the set $[\perp,b]_{L_1}$ which obviously has a least element (since $\perp$
is the least element of both $L_1$ and $L_2$ and therefore $\perp \in [\perp,b]_{L_1}$).
Let $S$ be a non-empty subset of $[\perp,b]_{L_1}$. Since $L_1$ is a complete
lattice, $\bigvee_{L_1} S \in L_1$. It suffices to show that
$\bigvee_{L_1} S \in [\perp,b]_{L_1}$. Since $S \subseteq [\perp,b]_{L_1}$,
for every $x \in S$ it holds $x \leq b$. By the Interlattice Lub Property,
$\bigvee_{L_1} S \leq b$, and therefore $\bigvee_{L_1} S \in [\perp,b]_{L_1}$.

The proof for the case of $[a,\top]_{L_2}$ is symmetrical and uses the
Interlattice Glb Property instead.
\end{proof}

The following proposition corresponds to Proposition 3.3 in~\cite{DMT04}:
\begin{proposition}
Let $(a,b) \in L_1 \otimes L_2$ and assume that $(a,b)$ is
$A$-reliable. Then, for every $x \in [\perp,b]_{L_1}$,
it holds $\perp \leq A(x,b)_1 \leq b$. Moreover, for every $x \in [a,\top]_{L_2}$,
it holds $a\leq A(a,x)_2 \leq \top$.
\end{proposition}
\begin{proof}
Define $a^* = \textit{lub}_{L_1}\{y \in L_1 \mid y \leq b\}$. By the fact that $a\leq b$ and the
definition of $a^*$, we get that $a \leq a^*$. By the Interlattice Lub Property we get that
$a^* \leq b$ and therefore $(a^*,b) \in L_1 \otimes L_2$. Moreover, $(x,b) \preceq (a^*,b)$.
Due to the $\preceq$-monotonicity of $A$ we have $A(x,b) \preceq A(a^*,b)$, and therefore
$A(x,b)_1 \leq A(a^*,b)_1$. Then:
\[
\begin{array}{llll}
A(a^*,b)_1                & \leq & A(a^*,b)_2   & (\mbox{Consistency of $A$})\\
                          & \leq & A(a,b)_2     & (\mbox{$a\leq a^*$ and $A$ is $\preceq$-monotone})\\
                          & \leq & b            & (\mbox{$A$-reliability})
\end{array}
\]

For the second part of the proof, define $b^* = \textit{glb}_{L_2}\{y \in L_2 \mid y \geq a\}$.
By the fact that $b \geq a$ and the definition of $b^*$, we get that $b^*\leq b$. By the Interlattice
Glb Property we get that $b^* \geq a$ and therefore $(a,b^*) \in L_1 \otimes L_2$.
Moreover, $(a,x) \preceq (a,b^*)$. Due to the $\preceq$-monotonicity of $A$ we have
$A(a,x) \preceq A(a,b^*)$, and therefore $A(a,x)_2 \geq A(a,b^*)_2$. Then:
\[
\begin{array}{llll}
A(a,b^*)_2                 & \geq & A(a,b^*)_1 & (\mbox{Consistency of $A$})\\
                           & \geq & A(a,b)_1   & (\mbox{$b^* \leq b$ and $A$ is $\preceq$-monotone})\\
                           & \geq & a          & (\mbox{$A$-reliability})
\end{array}
\]
This completes the proof of the proposition.
\end{proof}

The above proposition implies that for every $A$-reliable pair $(a,b)$, the restriction
of $A(.,b)_1$ to $[\perp,b]_{L_1}$ and the restriction of $A(a,.)_2$ to $[a,\top]_{L_2}$
are in fact operators (namely functions $[\perp,b]_{L_1} \rightarrow [\perp,b]_{L_1}$
and $[a,\top]_{L_2} \rightarrow [a,\top]_{L_2}$) on these intervals. Since by
Proposition~\ref{interals-complete-lattices} we know that $([\perp,b]_{L_1},\leq)$ and $([a,\top]_{L_2},\leq)$
are complete lattices, the operators $A(\cdot,b)_1$ and $A(a,\cdot)_2$ have least fixpoints in the
corresponding lattices. We define:
$$b^{\downarrow} = \textit{lfp}(A(\cdot,b)_1)$$
and
$$a^{\uparrow} = \textit{lfp}(A(a,\cdot)_2)$$
In the following, we will call the function mapping the $A$-reliable pair $(a,b)$ to
$(b^{\downarrow},a^{\uparrow})$, the {\em stable revision operator} for the
approximating operator $A$. We will denote this mapping by ${\cal C}_A$, namely:
$${\cal C}_A(x,y) = (y^{\downarrow},x^{\uparrow}) = (\textit{lfp}(A(\cdot,y)_1), \textit{lfp}(A(x,\cdot)_2))$$
We have the following proposition, which corresponds to Proposition 3.6 of~\cite{DMT04}:
\begin{proposition}\label{Corresponding_to_DMT_Prop_3.6}
Let $A\in \textit{Appx}(L_1\otimes L_2)$. For every $A$-reliable pair $(a,b)$,
$b^{\downarrow} \leq b$, $a \leq a^{\uparrow} \leq b$, and $(b^{\downarrow},a^{\uparrow})\in L_1 \otimes L_2$.
\end{proposition}
\begin{proof}
The inequalities $b^{\downarrow} \leq b$ and $a\leq a^{\uparrow}$ follow from the definition of
the stable revision operator. By the $A$-reliability of $(a,b)$ we have $A(a,b)_2 \leq b$ and
therefore $b$ is a pre-fixpoint of $A(a,\cdot)_2$. Since $a^{\uparrow}$ is the least pre-fixpoint
of $A(a,\cdot)_2$, we conclude that $a^{\uparrow} \leq b$.

Let $a^*=\textit{lub}_{L_1}\{x \in L_1 \mid x \leq a^{\uparrow}\}$.
Since $a\in \{x \in L_1 \mid x \leq a^{\uparrow}\}$ and since $a^*$
is the $\textit{lub}$ of this set, it holds $a \leq a^*$. Moreover,
notice that $a^*$ is in the domain of $A(\cdot,b)_1$ because
(by the Interlattice Lub Property) $a^* \leq a^{\uparrow}$, and since
$a^{\uparrow} \leq b$ we get $a^* \leq b$. We have:
\[
\begin{array}{llll}
   A(a^*,b)_1 & \leq & A(a^*,a^{\uparrow})_1 & (\mbox{$A$ is $\preceq$-monotonic})\\
              & \leq & A(a^*,a^{\uparrow})_2 & (\mbox{$A$ is consistent})\\
              & \leq & A(a,a^{\uparrow})_2   & (\mbox{$A$ is $\preceq$-monotonic})\\
              & =    & a^{\uparrow}          & (\mbox{$a^{\uparrow}$ fixpoint of $A(a,\cdot)_2$})
\end{array}
\]
Consequently, $A(a^*,b)_1 \leq a^{\uparrow}$ and therefore $A(a^*,b)_1 \in \{x \in L_1 \mid x \leq a^{\uparrow}\}$.
But $a^*=\textit{lub}_{L_1}\{x \in L_1 \mid x \leq a^{\uparrow}\}$ and therefore $A(a^*,b)_1 \leq a^*$.
It follows that $a^*$ is a pre-fixpoint of the operator $A(\cdot,b)_1$. Thus,
$b^{\downarrow} = \textit{lfp}(A(\cdot,b)_1) \leq a^* \leq a^{\uparrow}$.
\end{proof}
\begin{definition}
An $A$-reliable approximation $(a,b)$ is $A$-prudent if $a\leq b^{\downarrow}$.
\end{definition}
\begin{proposition}
Let $A\in \textit{Appx}(L_1\otimes L_2)$ and let $(a,b) \in L_1\otimes L_2$ be $A$-prudent.
Then, $(a,b) \preceq (b^{\downarrow},a^{\uparrow})$ and $(b^{\downarrow},a^{\uparrow})$ is $A$-prudent.
\end{proposition}
\begin{proof}
By Proposition~\ref{Corresponding_to_DMT_Prop_3.6}, it holds $b^{\downarrow} \leq b$, $a \leq a^{\uparrow}$
and $a^{\uparrow} \leq b$. Since $(a,b)$ is $A$-prudent, we get $(a,b) \preceq (b^{\downarrow},a^{\uparrow})$.

Notice now that by the $\preceq$ monotonicity of $A$ we get that
$b^{\downarrow} = A(b^{\downarrow},b)_1 \leq A(b^{\downarrow},a^{\uparrow})_1$ and
$a^{\uparrow} = A(a,a^{\uparrow})_2 \geq A(b^{\downarrow},a^{\uparrow})_2$.
This implies that $(b^{\downarrow},a^{\uparrow})$ is $A$-reliable.

Observe now that since $a^{\uparrow} \leq b$ and $A$ is $\preceq$-monotonic,
it holds that for every $x \in [\perp,a^{\uparrow}]_{L_1}$, $A(x,b)_1 \leq A(x,a^{\uparrow})_1$.
Therefore, each pre-fixpoint of $A(\cdot,a^{\uparrow})_1$ is a pre-fixpoint
of $A(\cdot,b)_1$. By the proof of Proposition~\ref{Corresponding_to_DMT_Prop_3.6}
we have that $A(a^*,a^{\uparrow})_1 \leq a^{\uparrow}$, and by the definition of
$a^*$ in that same proof, it follows that $A(a^*,a^{\uparrow})_1 \leq a^*$.
Therefore the set of pre-fixpoints of $A(\cdot,a^{\uparrow})_1$ is non-empty. Consequently,
$b^{\downarrow} = \textit{lfp}(A(\cdot,b)_1) \leq  \textit{lfp}(A(\cdot,a^{\uparrow})_1) =
(a^{\uparrow})^{\downarrow}$, and therefore $(b^{\downarrow},a^{\uparrow})$ is $A$-prudent.
\end{proof}

The following proposition (corresponding to Proposition 2.3 in~\cite{DMT04})
now requires in its proof the Interlattice Lub Property.
\begin{proposition}
Let $\{(a_{\kappa},b_{\kappa})\}_{\kappa<\lambda}$, where $\lambda$ is an ordinal,
be a chain in $L_1\otimes L_2$ ordered by the relation $\preceq$. Then:
\begin{enumerate}
\item $\bigvee_{L_1}\{a_{\kappa}\mid \kappa<\lambda\} \leq \bigwedge_{L_2}\{b_{\kappa}\mid \kappa<\lambda\}$.
\item The least upper bound of the chain with respect to $\preceq$ exists,
      and is equal to $(\bigvee_{L_1}\{a_{\kappa}\mid \kappa<\lambda\}, \bigwedge_{L_2}\{b_{\kappa}\mid \kappa<\lambda\})$.
\end{enumerate}
\begin{proof}
We demonstrate the first statement; the proof of the second part is easy and omitted.
For the proof of the first part, notice that since the chain is ordered by $\preceq$,
$\bigwedge_{L_2}\{b_{\kappa}\mid \kappa<\lambda\} = b_0$. Moreover, for every $\kappa <\lambda$
it holds $a_\kappa \leq b_\kappa$ because $(a_\kappa, b_\kappa)\in L_1 \otimes L_2$;
since $b_\kappa \leq b_0$, it is $a_\kappa\leq b_0$ for all $\kappa < \lambda$.
By the Interlattice Lub Property, we get
$\bigvee_{L_1}\{a_{\kappa}\mid \kappa<\lambda\} \leq b_0 = \bigwedge_{L_2}\{b_{\kappa}\mid \kappa<\lambda\}$.
\end{proof}
\end{proposition}
The following proposition (corresponding to Proposition 3.10 in~\cite{DMT04})
and the subsequent theorem (corresponding to Theorem 3.11 in~\cite{DMT04}) have
identical proofs to the ones given in~\cite{DMT04} (the only difference being
that our underlying domain is $L_1\otimes L_2$):
\begin{proposition}
Let $A\in \textit{Appx}(L_1\otimes L_2)$ and let $\{(a_{\kappa},b_{\kappa})\}_{\kappa<\lambda}$, where $\lambda$ is an ordinal,
be a chain of $A$-prudent pairs from $L_1\otimes L_2$. Then, $\bigvee_{\preceq}  \{(a_{\kappa},b_{\kappa})\}_{\kappa<\lambda}$,
is $A$-prudent.
\end{proposition}
\begin{theorem}\label{central-theorem}
Let $A\in \textit{Appx}(L_1\otimes L_2)$. The set of $A$-prudent  elements of
$L_1\otimes L_2$ is a chain-complete poset under $\preceq$ with least element
$(\perp,\top)$. The stable revision operator is a well-defined, increasing
and monotone operator in this poset, and therefore it has a least fixpoint
which is $A$-prudent and can be obtained as the limit of the following sequence:
\[
\begin{array}{llll}
(a_0,b_0)  & = & (\perp,\top) & \\
(a_{\lambda+1}, b_{\lambda+1}) &  = & {\cal C}_A(a_{\lambda}, b_{\lambda}) & \\
(a_{\lambda},b_{\lambda}) & = & \bigvee_{\preceq}\{(a_{\kappa},b_{\kappa}): \kappa <\lambda \} & \mbox{for limit ordinals $\lambda$}
\end{array}
\]
\end{theorem}%

The proof of the following theorem is also a straightforward generalization of the
proof of Theorem 19 in~\cite{DMT00}:
\begin{theorem}\label{minimal-fixpoint}
Every fixpoint of the stable revision operator ${\cal C}_A$ is a $\leq$-minimal
pre-fixpoint of $A$.
\end{theorem}

\section{$\!\!\!$: Proofs of Section~\ref{well_founded}}\label{appendix-of-section-6b}
Before providing the proofs of the results of Section~\ref{well_founded}, we notice that
Proposition~\ref{tau-monotonic} extends to the case of Herbrand interpretations as follows:
\begin{proposition}\label{tau-monotonic-for-interpretations}
Let $\mathsf{P}$ be a program. Then, for every ${\cal I},{\cal J} \in {\cal H}_{\mathsf{P}}$
and for every $(I_1,J_1),(I_2,J_2) \in ({\cal H}^{\mathsf{ma}}_\mathsf{P} \otimes  {\cal H}^{\mathsf{am}}_\mathsf{P})$,
the following statements hold:
  \begin{enumerate}
  \item $\tau({\cal I}) \in  ({\cal H}^{\mathsf{ma}}_\mathsf{P} \otimes  {\cal H}^{\mathsf{am}}_\mathsf{P})$
        and $\tau^{-1}(I_1,J_1) \in {\cal H}_{\mathsf{P}}$.
  \item If ${\cal I} \preceq {\cal J}$ then $\tau({\cal I}) \preceq \tau({\cal J})$.
  \item If ${\cal I} \leq {\cal J}$ then $\tau({\cal I}) \leq \tau({\cal J})$.
  \item If $(I_1,J_1) \preceq (I_2,J_2)$ then $\tau^{-1}(I_1,J_1) \preceq \tau^{-1}(I_2,J_2)$.
  \item If $(I_1,J_1) \leq (I_2,J_2)$ then $\tau^{-1}(I_1,J_1) \leq \tau^{-1}(I_2,J_2)$.
\end{enumerate}
\end{proposition}
\begingroup
\def\thelemma{\ref{T_P_is_Fitting_monotonic}}
\begin{lemma}
Let $\mathsf{P}$ be a program and let  $(I_1, J_1), (I_2, J_2) \in {\cal H}^{\mathsf{ma}}_\mathsf{P} \otimes  {\cal H}^{\mathsf{am}}_\mathsf{P}$.
If $(I_1, J_1) \preceq (I_2, J_2)$ then $T_\mathsf{P}(I_1, J_1) \preceq T_\mathsf{P}(I_2,J_2)$.
\end{lemma}
\addtocounter{lemma}{-1}
\endgroup
\begin{proof}
It follows directly from the definition of $\Psi_{\mathsf{P}}$ together with Lemma~\ref{expressions_fitting} and
Corollary~\ref{prop-lub-preceq} in~\ref{appendix-of-section-4} that $\Psi_\mathsf{P}$ is $\preceq$-monotonic.
It follows from Proposition~\ref{tau-monotonic-for-interpretations} that $\tau^{-1}(I_1,J_1) \preceq \tau^{-1}(I_2,J_2)$.
Since $\Psi_\mathsf{P}$ is $\preceq$-monotonic we get $\Psi_\mathsf{P}(\tau^{-1}(I_1,J_1)) \preceq \Psi_\mathsf{P}(\tau^{-1}(I_2,J_2))$.
By applying again Proposition~\ref{tau-monotonic} we have that $T_\mathsf{P}(I_1,J_1) \preceq T_\mathsf{P}(I_2,J_2)$
that concludes the proof.
\end{proof}

\begingroup
\def\thelemma{\ref{tp-fixpoint-model}}
\begin{lemma}
Let $\mathsf{P}$ be a program. If $(I, J)\in {\cal H}^{\mathsf{ma}}_\mathsf{P} \otimes  {\cal H}^{\mathsf{am}}_\mathsf{P}$
is a pre-fixpoint of $T_\mathsf{P}$ then $\tau^{-1}(I, J)$ is a model of $\mathsf{P}$.
\end{lemma}
\addtocounter{lemma}{-1}
\endgroup
\begin{proof}
From the definition of $T_{\mathsf{P}}$ and using the fact that $(I, J)$ is a pre-fixpoint of $T_\mathsf{P}$,
it follows  that $\tau(\Psi_\mathsf{P}(\tau^{-1}(I,J))) = T_\mathsf{P}(I, J) \leq (I,J)$.
By applying $\tau^{-1}$ to both sides of the statement and using
Proposition~\ref{tau-monotonic-for-interpretations} we get that
$\tau^{-1}(\tau(\Psi_\mathsf{P}(\tau^{-1}(I,J)))) \leq \tau^{-1}(I, J)$
which gives $\Psi_\mathsf{P}(\tau^{-1}(I,J)) \leq \tau^{-1}(I,J)$.
From the definition of $\Psi_\mathsf{P}$ and the definition of model,
it follows that $\tau^{-1}(I,J)$ is model of $\mathsf{P}$.
\end{proof}
\begingroup
\def\thelemma{\ref{is_a_prefixpoint}}
\begin{lemma}
Let ${\cal M}\in {\cal H}_{\mathsf{P}}$ be a model of $\mathsf{P}$. Then, $\tau({\cal M})$ is a pre-fixpoint of $T_\mathsf{P}$.
\end{lemma}
\addtocounter{lemma}{-1}
\endgroup
\begin{proof}
By the definition of $\Psi_{\mathsf{P}}$ we have that for every predicate constant $\mathsf{p}$ in $\mathsf{P}$,
$\Psi_\mathsf{P}({\cal M})(\mathsf{p}) = \bigvee_{\leq}\{ \lsem \mathsf{E} \rsem({\cal M}) \mid (\mathsf{p} \leftarrow \mathsf{E}) \in \mathsf{P}\}$.
Since ${\cal M}$ is a model of $\mathsf{P}$ it follows that $\lsem \mathsf{E} \rsem({\cal M}) \leq {\cal M}(\mathsf{p})$
for every clause $\mathsf{p} \leftarrow \mathsf{E}$ in $\mathsf{P}$, i.e., ${\cal M}(\mathsf{p})$ is an
upper bound of the set $\{ \lsem \mathsf{E} \rsem({\cal M}) \mid (\mathsf{p} \leftarrow \mathsf{E}) \in \mathsf{P}\}$.
Therefore, $\bigvee_{\leq}\{ \lsem \mathsf{E} \rsem({\cal M}) \mid (\mathsf{p} \leftarrow \mathsf{E}) \in \mathsf{P}\} \leq {\cal M}(\mathsf{p})$,
which implies that $\Psi_\mathsf{P}({\cal M}) \leq {\cal M}$. By Proposition~\ref{tau-monotonic-for-interpretations}
it follows that $\tau(\Psi_\mathsf{P}({\cal M})) \leq \tau({\cal M})$.
Moreover, by the definition of $T_\mathsf{P}$ and Proposition~\ref{tau-monotonic-for-interpretations} we have that
$T_\mathsf{P}(\tau({\cal M})) = \tau(\Psi_\mathsf{P}(\tau^{-1}(\tau({\cal M})))) = \tau(\Psi_\mathsf{P}(M)) \leq \tau({\cal M})$,
and therefore $\tau({\cal M})$ is a pre-fixpoint of $T_\mathsf{P}$.
\end{proof}
In order to establish Theorem~\ref{minimal-model} that follows, we need the following lemma:
\begin{lemma}\label{minimal_prefixpoint}
Let $\mathsf{P}$ be a program. If $(I, J)\in {\cal H}^{\mathsf{ma}}_\mathsf{P} \otimes  {\cal H}^{\mathsf{am}}_\mathsf{P}$
is a minimal pre-fixpoint of $T_\mathsf{P}$ then $\tau^{-1}(I, J)$ is a minimal model of $\mathsf{P}$.
\end{lemma}
\begin{proof}
Let ${\cal M} = \tau^{-1}(I, J)$. By Lemma~\ref{tp-fixpoint-model}, ${\cal M}$ is a model of $\mathsf{P}$.
Assume there exists a model ${\cal N}\in {\cal H}_{\mathsf{P}}$ of $\mathsf{P}$ such that ${\cal N} \leq {\cal M}$.
Applying $\tau$ to both sides and using Proposition~\ref{tau-monotonic-for-interpretations} we get that
$\tau({\cal N}) \leq \tau({\cal M})$. By Lemma~\ref{is_a_prefixpoint}, $\tau({\cal N})$ is a pre-fixpoint
of $T_\mathsf{P}$ and since $\tau({\cal M})=(I, J)$ is a minimal pre-fixpoint of  $T_\mathsf{P}$,
we get that $\tau({\cal N}) = \tau({\cal M})$. Applying $\tau^{-1}$ to both sides, we get
${\cal N} ={\cal M}$.
\end{proof}
\begingroup
\def\thetheorem{\ref{minimal-model}}
\begin{theorem}
Let $\mathsf{P}$ be a program. Then, ${\cal M}_{\mathsf{P}}$ is a $\leq$-minimal model of $\mathsf{P}$.
\end{theorem}
\addtocounter{theorem}{-1}
\endgroup
\begin{proof}
By Theorem~\ref{minimal-fixpoint} (see~\ref{appendix-of-section-6a}) every fixpoint of ${\cal C}_{T_{\mathsf{P}}}$ is a
{\em minimal pre-fixpoint} of $T_{\mathsf{P}}$. Since by Theorem~\ref{iterative_definition_of_wfm}
$(I_{\delta},J_{\delta})=\tau({\cal M}_{\mathsf{P}})$ is a fixpoint of ${\cal C}_{T_{\mathsf{P}}}$,
$\tau({\cal M}_{\mathsf{P}})$ is a minimal pre-fixpoint of $T_{\mathsf{P}}$. By Lemma~\ref{minimal_prefixpoint},
$\tau^{-1}(\tau({\cal M}_{\mathsf{P}})) = {\cal M}_{\mathsf{P}}$ is a minimal model of $\mathsf{P}$.
\end{proof}
\begingroup
\def\thetheorem{\ref{backwards-compatible}}
\begin{theorem}
For every propositional program $\mathsf{P}$, ${\cal M}_{\mathsf{P}}$ coincides with
the well-founded model of $\mathsf{P}$.
\end{theorem}
\addtocounter{theorem}{-1}
\endgroup
\begin{proof}
In~\cite{DMT04}[Section 6, pages 107-108], the well-founded semantics of propositional
logic programs (allowing arbitrary nesting of conjunction, disjunction and negation
in clause bodies) is derived. By a careful inspection of the steps used in the above
reference, it can be seen that the construction given therein is a special case of the
technique used in the present paper.
\end{proof}

\section{$\!\!\!$: The Model ${\cal M}_{\mathsf{P}}$ for an Example Program}\label{appendix-of-section-6c}
Consider the following program $\mathsf{P}$ which is a simplified non-recursive version
of a program taken from~\cite{RondogiannisS17}. Initially we use a Prolog-like syntax:
\[
\begin{array}{l}
\mbox{\tt s(Q,V) $\leftarrow$ Q(V)}\\
\mbox{\tt p(R) $\leftarrow$ R}\\
\mbox{\tt q(R) $\leftarrow$ $\mysim$ w(R)}\\
\mbox{\tt w(R) $\leftarrow$ $\mysim$ R}
\end{array}
\]
In the above example, the type of {\tt p}, {\tt q} and {\tt w} is $o\rightarrow o$, and the
type of {\tt s} is $(o\rightarrow o)\rightarrow o \rightarrow o$. In ${\cal HOL}$ notation
the program can be written as follows:
\[
\begin{array}{l}
\mbox{\tt s  $\leftarrow$ $\lambda$Q.$\lambda$V.(Q V)}\\
\mbox{\tt p  $\leftarrow$ $\lambda$R.R}\\
\mbox{\tt q  $\leftarrow$ $\lambda$R.$\mysim$ (w R)}\\
\mbox{\tt w  $\leftarrow$ $\lambda$R.($\mysim$ R)}
\end{array}
\]
Notice now that the bodies of the clauses of {\tt s}, {\tt q} and {\tt w} do not involve
other predicate constants, and therefore the calculation of their meaning can be performed
in a more direct way. On the other hand, the body of the clause concerning {\tt q} involves the
predicate constant {\tt w}, and therefore the calculation of the meaning of {\tt q} is more involved.

The first approximation to the well-founded model of $\mathsf{P}$ is the pair
$(I_0,J_0) = (\perp,\top)$ (see Theorem~\ref{iterative_definition_of_wfm}). Consider
now $(I_1,J_1)$. We have:
$$I_{1} = \textit{lfp}([T_{\mathsf{P}}(\cdot,\top)]_1) = \textit{lfp}([\tau(\Psi_\mathsf{P}(\tau^{-1}(\cdot,\top)))]_1)$$
and
$$J_{1} = \textit{lfp}([T_{\mathsf{P}}(\perp,\cdot)]_2) = \textit{lfp}([\tau(\Psi_\mathsf{P}(\tau^{-1}(\perp,\cdot)))]_2)$$
where, as discussed in~\ref{appendix-of-section-6a}, the $\textit{lfp}$ in the case of $I_1$ is the least upper bound of
the sequence $I_1^0,I_1^1,\ldots$, defined as follows:
\[
\begin{array}{lll}
I_1^0  & = & [\tau(\Psi_\mathsf{P}(\tau^{-1}(\perp,\top)))]_1\\
I_1^1  & = & [\tau(\Psi_\mathsf{P}(\tau^{-1}(I_1^0,\top)))]_1\\
       & \cdots & \\
I_1^{\alpha+1}  & = & [\tau(\Psi_\mathsf{P}(\tau^{-1}(I_1^{\alpha},\top)))]_1\\
       & \cdots & \\
\end{array}
\]
and the $\textit{lfp}$ in the case of $J_1$ is the least upper bound of the sequence $J_1^0,J_1^1,\ldots$, defined
as follows:
\[
\begin{array}{lll}
J_1^0  & = & [\tau(\Psi_\mathsf{P}(\tau^{-1}(\perp,\perp)))]_2\\
J_1^1  & = & [\tau(\Psi_\mathsf{P}(\tau^{-1}(\perp, J_1^0)))]_2\\
       & \cdots & \\
J_1^{\alpha+1}  & = & [\tau(\Psi_\mathsf{P}(\tau^{-1}(\perp,J_1^{\alpha})))]_2\\
       & \cdots & \\
\end{array}
\]
For the predicate constant {\tt w} we have:
\[
\begin{array}{l}
I_1^0({\tt w})   =  [\tau(\Psi_\mathsf{P}(\tau^{-1}(\perp,\top)))]_1({\tt w})  =
                      [\tau(\lsem \mbox{\tt $\lambda$R.$\mysim$ R}\rsem(\tau^{-1}(\perp,\top)))]_1  =
                      [\tau(\lambda v.v^{-1})]_1\\
I_1^1({\tt w})   =  [\tau(\Psi_\mathsf{P}(\tau^{-1}(I_1^0,\top)))]_1({\tt w}) =
                      [\tau(\lsem \mbox{\tt $\lambda$R.$\mysim$ R}\rsem(\tau^{-1}(I_1^0,\top)))]_1  =
                      [\tau(\lambda v.v^{-1})]_1\\
\hspace{2cm}\cdots  \\
I_1^{\alpha+1}({\tt w})   =  [\tau(\Psi_\mathsf{P}(\tau^{-1}(I_1^{\alpha},\top)))]_1({\tt w}) =
                      [\tau(\lsem \mbox{\tt $\lambda$R.$\mysim$ R}\rsem(\tau^{-1}(I_1^{\alpha},\top)))]_1  =
                      [\tau(\lambda v.v^{-1})]_1\\
\hspace{2cm}\cdots
\end{array}
\]
Similarly, we can show that for every ordinal $\alpha$, $J_1^{\alpha}({\tt w})= [\tau(\lambda v.v^{-1})]_2$. The above
imply that ${\cal M}_{\mathsf{P}}({\tt w}) = \lambda v.v^{-1}$. In other words, the denotation of {\tt w}
is the {\em not} function over our 3-valued truth domain. In a similar way, it follows that
${\cal M}_{\mathsf{P}}({\tt p}) = \lambda v.v$. In other words, the denotation of {\tt p} is the identity
function over our 3-valued domain.

Consider now the predicate constant {\tt q}. We have:
\[
\begin{array}{l}
I_1^0({\tt q})   =  [\tau(\Psi_\mathsf{P}(\tau^{-1}(\perp,\top)))]_1({\tt q}) =
                      [\tau(\lsem \mbox{\tt $\lambda$R.$\mysim$(w R)}\rsem(\tau^{-1}(\perp,\top)))]_1 =
                      [\tau(\lambda v.\textit{undef})]_1\\
I_1^1({\tt q})  = [\tau(\Psi_\mathsf{P}(\tau^{-1}(I_1^0,\top)))]_1({\tt q}) =
                      [\tau(\lsem \mbox{\tt $\lambda$R.$\mysim$(w R)}\rsem(\tau^{-1}(I_1^0,\top)))]_1 =
                      [\tau(f)]_1\\
\hspace{2cm}\cdots \\
I_1^{\alpha+1}({\tt q})  = [\tau(\Psi_\mathsf{P}(\tau^{-1}(I_1^{\alpha},\top)))]_1({\tt q}) =
                      [\tau(\lsem \mbox{\tt $\lambda$R.$\mysim$(w R)}\rsem(\tau^{-1}(I_1^{\alpha},\top)))]_1 =
                      [\tau(f)]_1\\
\hspace{2cm}\cdots
\end{array}
\]
where $f$ is the function such that $f(\textit{true})=f(\textit{undef})=\textit{undef}$ and $f(\textit{false})=\textit{false}$.
Similarly, we have:
\[
\begin{array}{l}
J_1^0({\tt q})  = [\tau(\Psi_\mathsf{P}(\tau^{-1}(\perp,\perp)))]_2({\tt q}) =
                      [\tau(\lsem \mbox{\tt $\lambda$R.$\mysim$(w R)}\rsem(\tau^{-1}(\perp,\perp)))]_2 =
                      [\tau(\lambda v.\textit{true})]_2\\
J_1^1({\tt q})  = [\tau(\Psi_\mathsf{P}(\tau^{-1}(\perp,J^0_1)))]_2({\tt q}) =
                      [\tau(\lsem \mbox{\tt $\lambda$R.$\mysim$(w R)}\rsem(\tau^{-1}(\perp,J_1^0)))]_2 =
                      [\tau(g)]_2\\
\hspace{2cm}\cdots  \\
J_1^{\alpha+1}({\tt q})  = [\tau(\Psi_\mathsf{P}(\tau^{-1}(\perp,J^{\alpha}_1)))]_2({\tt q}) =
                      [\tau(\lsem \mbox{\tt $\lambda$R.$\mysim$(w R)}\rsem(\tau^{-1}(\perp,J_1^{\alpha})))]_2 =
                      [\tau(g)]_2\\
\hspace{2cm}\cdots
\end{array}
\]
where $g$ is the function such that $g(\textit{false})=g(\textit{undef})=\textit{undef}$ and $g(\textit{true})=\textit{true}$.

Consider now $(I_2,J_2)$. We have:
$$I_{2} = \textit{lfp}([T_{\mathsf{P}}(\cdot,J_1)]_1) = \textit{lfp}([\tau(\Psi_\mathsf{P}(\tau^{-1}(\cdot,J_1)))]_1)$$
and
$$J_{2} = \textit{lfp}([T_{\mathsf{P}}(I_1,\cdot)]_2) = \textit{lfp}([\tau(\Psi_\mathsf{P}(\tau^{-1}(I_1,\cdot)))]_2)$$
where the $\textit{lfp}$ in the case of $I_2$ is the least upper bound of the sequence $I_2^0,I_2^1,\ldots$ defined
as follows:
\[
\begin{array}{lll}
I_2^0  & = & [\tau(\Psi_\mathsf{P}(\tau^{-1}(\perp,J_1)))]_1\\
I_2^1  & = & [\tau(\Psi_\mathsf{P}(\tau^{-1}(I_2^0,J_1)))]_1\\
       & \cdots & \\
I_2^{\alpha+1}  & = & [\tau(\Psi_\mathsf{P}(\tau^{-1}(I_2^{\alpha},J_1)))]_1\\
       & \cdots & \\
\end{array}
\]
and the $\textit{lfp}$ in the case of $J_2$ is the least upper bound of the sequence $J_2^0,J_2^1,\ldots$ defined
as follows:
\[
\begin{array}{lll}
J_2^0  & = & [\tau(\Psi_\mathsf{P}(\tau^{-1}(I_1,I_1^{*})))]_2\\
J_2^1  & = & [\tau(\Psi_\mathsf{P}(\tau^{-1}(I_1, J_2^0)))]_2\\
       & \cdots & \\
J_2^{\alpha+1}  & = & [\tau(\Psi_\mathsf{P}(\tau^{-1}(I_1,J_2^{\alpha})))]_2\\
       & \cdots & \\
\end{array}
\]
where $I_1^*$ is the least interpretation in ${\cal H}_{\mathsf{P}}^{\mathsf{am}}$
such that $I_1 \leq I_1^*$ (namely, the bottom antimonotone-monotone element of the
interval $[I_1,\perp]$, see the construction in~\ref{appendix-of-section-6a}).

Consider again the predicate constant {\tt q}. We have:
\[
\begin{array}{l}
I_2^0({\tt q})  = [\tau(\Psi_\mathsf{P}(\tau^{-1}(\perp,J_1)))]_1({\tt q}) =
                      [\tau(\lsem \mbox{\tt $\lambda$R.$\mysim$(w R)}\rsem(\tau^{-1}(\perp,J_1)))]_1\\
I_2^1({\tt q})  = [\tau(\Psi_\mathsf{P}(\tau^{-1}(I_2^0,J_1)))]_1({\tt q}) =
                      [\tau(\lsem \mbox{\tt $\lambda$R.$\mysim$(w R)}\rsem(\tau^{-1}(I_2^0,J_1)))]_1 =
                      [\tau(\lambda v.v)]_1\\
\hspace{2cm}\cdots \\
I_2^{\alpha+1}({\tt q})  = [\tau(\Psi_\mathsf{P}(\tau^{-1}(I_2^{\alpha},J_1)))]_1({\tt q}) =
                      [\tau(\lsem \mbox{\tt $\lambda$R.$\mysim$(w R)}\rsem(\tau^{-1}(I_2^{\alpha},J_1)))]_1 =
                      [\tau(\lambda v.v)]_1\\
\hspace{2cm}\cdots
\end{array}
\]
because for all ordinals $\alpha$, $I_2^{\alpha}({\tt w}) = [\tau(\lambda v.v^{-1})]_1$ and
$J_1({\tt w}) = [\tau(\lambda v.v^{-1})]_2$. Similarly, we have:
\[
\begin{array}{l}
J_2^0({\tt q})  = [\tau(\Psi_\mathsf{P}(\tau^{-1}(I_1,I_1^*)))]_2({\tt q})  =
                      [\tau(\lsem \mbox{\tt $\lambda$R.$\mysim$(w R)}\rsem(\tau^{-1}(I_1,I_1^*)))]_2 \\
J_2^1({\tt q})  =  [\tau(\Psi_\mathsf{P}(\tau^{-1}(I_1,J^0_2)))]_2({\tt q}) =
                      [\tau(\lsem \mbox{\tt $\lambda$R.$\mysim$(w R)}\rsem(\tau^{-1}(I_1,J^0_2)))]_2  =
                      [\tau(\lambda v.v)]_2\\
\hspace{2cm}\cdots\\
J_2^{\alpha+1}({\tt q})  =   [\tau(\Psi_\mathsf{P}(\tau^{-1}(I_1,J^{\alpha}_2)))]_2({\tt q}) =
                      [\tau(\lsem \mbox{\tt $\lambda$R.$\mysim$(w R)}\rsem(\tau^{-1}(I_1,J_2^{\alpha})))]_2  =
                      [\tau(\lambda v.v)]_2\\
\hspace{2cm}\cdots
\end{array}
\]
because $I_1({\tt w}) = [\tau(\lambda v.v^{-1})]_1$ and for all ordinals $\alpha$, $J_2^\alpha({\tt w}) = [\tau(\lambda v.v^{-1})]_2$.
The above imply that ${\cal M}_{\mathsf{P}}({\tt q}) = \lambda v.v$. In other words, the denotation of {\tt q}
is the identity function over our 3-valued truth domain. Notice that despite their different definitions,
{\tt p} and {\tt q} denote the same 3-valued relation (in some sense, the two negations in the definition
of {\tt q} cancel each other).

Finally, consider the predicate constant {\tt s}. We have:
\[
\begin{array}{l}
I_1^0({\tt s})   =
                      [\tau(\lsem \mbox{\tt $\lambda$Q.$\lambda$V.(Q V)}\rsem(\tau^{-1}(\perp,J_1)))]_1  =
                      [\tau(\lambda q.\lambda v.(q\,  v))]_1\\
I_1^1({\tt s})   =
                      [\tau(\lsem \mbox{\tt $\lambda$Q.$\lambda$V.(Q V)}\rsem(\tau^{-1}(I_1^0,\top)))]_1  =
                      [\tau(\lambda q.\lambda v.(q\,  v))]_1\\
\hspace{2cm}\cdots  \\
I_1^{\alpha+1}({\tt s})   =
                      [\tau(\lsem \mbox{\tt $\lambda$Q.$\lambda$V.(Q V)}\rsem(\tau^{-1}(I_1^{\alpha},\top)))]_1  =
                      [\tau(\lambda q.\lambda v.(q\,  v))]_1\\
\hspace{2cm}\cdots
\end{array}
\]
and also:
\[
\begin{array}{l}
J_1^0({\tt s})   =
                      [\tau(\lsem \mbox{\tt $\lambda$Q.$\lambda$V.(Q V)}\rsem(\tau^{-1}(I_1,\perp)))]_2  =
                      [\tau(\lambda q.\lambda v.(q\,  v))]_2\\
J_1^1({\tt s})   =
                      [\tau(\lsem \mbox{\tt $\lambda$Q.$\lambda$V.(Q V)}\rsem(\tau^{-1}(\perp,J_1^0)))]_2  =
                      [\tau(\lambda q.\lambda v.(q\,  v))]_2\\
\hspace{2cm}\cdots  \\
J_1^{\alpha+1}({\tt s})   =
                      [\tau(\lsem \mbox{\tt $\lambda$Q.$\lambda$V.(Q V)}\rsem(\tau^{-1}(\perp,J_1^{\alpha})))]_2  =
                      [\tau(\lambda q.\lambda v.(q\,  v))]_2\\
\hspace{2cm}\cdots
\end{array}
\]
The above imply that ${\cal M}_{\mathsf{P}}({\tt s}) = \lambda q.\lambda v. (q\, v)$.

\end{document}